\documentclass[10pt]{article} 

\usepackage{fullpage}
\usepackage{enumitem}
\usepackage[page]{appendix}
\usepackage{amssymb}
\usepackage{mathrsfs} 
\usepackage{natbib}
\usepackage{algorithm}
\usepackage{algorithmic}
\usepackage{hyperref}
\usepackage{comment}
\usepackage{multirow}
\usepackage{amsthm}
\usepackage{mathtools}
\usepackage{epstopdf,xcolor}
\usepackage{bbm}
\usepackage[colorinlistoftodos,prependcaption]{todonotes}
\usepackage{dsfont}
\usepackage{pifont}
\usepackage{subcaption}
\usepackage{tikz, subcaption}
\usetikzlibrary{shapes, arrows, positioning}
\usepackage{authblk}
 \usepackage{listings}

\newtheorem{theorem}{Theorem}

\newtheorem{corollary}{Corollary}

\newtheorem{remark}{Remark}

\newtheorem{assumption}{Assumption}

\newcommand{\E}{\mathbb{E}}

\newcommand{\independent}{\perp\mkern-9.5mu\perp}

\usepackage[normalem]{ulem}
\newcommand{\stkout}[1]{\ifmmode\text{\sout{\ensuremath{#1}}}\else\sout{#1}\fi}

\title{\bf \Large Partial Identification of Causal Effects Using Proxy Variables\\~\\} 

\author[1]{AmirEmad Ghassami}
\author[2]{Yunshu Zhang}
\author[3]{Ilya Shpitser}
\author[2]{Eric Tchetgen Tchetgen}
\affil[1]{Department of Mathematics and Statistics, Boston University}
\affil[2]{Department of Statistics and Data Science, University of Pennsylvania}
\affil[3]{Department of Computer Science, Johns Hopkins University}

\date{\vspace{-0mm}First Version: April 10, 2023; Current Version: August 05, 2026\vspace{-0mm}}
\begin{document}
\maketitle

\begin{abstract}
Proximal causal inference is a framework for evaluating the causal effects in the presence of unmeasured confounding. For point identification, it leverages a pair of proxy variables to identify a bridge function that matches the dependence of potential outcomes or treatment variables on the hidden factors to corresponding functions of observed proxies. Unique identification requires that proxies are sufficiently relevant for hidden factors, a requirement that has previously been formalized as a completeness condition. However, completeness is not empirically testable, and although a bridge function may be well-defined in a given setting, lack of completeness, sometimes manifested by availability of a single type of proxy, may severely limit prospects for identification of a bridge function and thus a causal effect; therefore, potentially restricting the application of the framework. In this paper, we propose partial identification methods that do not require completeness and obviate the need for identification of a bridge function. We establish that proxies can be leveraged to obtain bounds on the causal effect even if available information does not suffice to identify either a bridge function or a corresponding causal effect of interest. Our bounds are non-smooth functionals of the underlying distribution. For inference, we employ LogSumExp approximations that yield smooth lower and upper bounds, and we derive the efficient influence functions of the resulting bound functionals which enable analytic variance estimation, while bootstrap confidence intervals remain available for regular plug-in implementations. We further establish analogous results in related settings where identification hinges upon hidden mediators for which proxies are available, however such proxies are not sufficiently rich for point identification of a bridge function or a corresponding causal effect of interest.\\

\end{abstract}

\newpage

\section{Introduction}
\label{sec:intro}

Evaluation of the causal effect of a certain treatment variable on an outcome variable of interest from purely observational data is the main focus in many scientific endeavors. One of the most common identification conditions for causal inference from observational data is that of conditional exchangeability. The assumption essentially presumes that the researcher has collected a sufficiently rich set of pre-treatment covariates, such that within the covariate strata, it is as if the treatment were assigned randomly. Unfortunately, this assumption is violated in many real-world settings as it essentially requires that there should not exist any unobserved common causes of the treatment-outcome relation (i.e., no unobserved confounders).

To address the challenge of unobserved confounders, the proximal causal inference framework was recently introduced by Tchetgen Tchetgen and colleagues \citep{tchetgen2020introduction,miao2018identifying}. This framework shows that identification of the causal effect of the treatment on outcome in the presence of unobserved confounders still is sometimes feasible provided that one has access to two types of proxy variables of the unobserved confounders, a treatment confounding proxy and an outcome confounding proxy that satisfy certain assumptions. The proximal causal inference framework was recently extended to several other challenging causal inference settings such longitudinal data \citep{10.1093/jrsssb/qkad020}, mediation analysis \citep{10.1093/biomet/asad015,ghassami2021hidmed}, outcome-dependent sampling \citep{LiQ}, network interference settings \citep{egami2023identification}, graphical causal models \citep{shpitser2023proximal}, and causal data fusion \citep{ghassami2022combining}. It has also recently been shown that under an alternative somewhat stronger set of assumptions with regards to proxy relevance, identification is sometimes possible using a single proxy variable \citep{tchetgen2023single}.

A key identification condition of proximal causal inference is an assumption that proxies are sufficiently relevant for hidden confounding factors so that there exist a so-called confounding bridge function, defined in terms of proxies, that matches the association between hidden factors and either the potential outcomes or the treatment variable; identification of which is an important step towards identification of causal effects via proxies. However, non-parametric identification of such a bridge function, when it exists, has to date involved a completeness condition which, roughly speaking, requires that variation in proxy variables reflects all sources of variation of unobserved confounders \citep{tchetgen2020introduction}. 
However, completeness is a strong condition which may significantly limit the researcher's choice of proxy variables, and hence the feasibility of proximal causal inference in important practical settings where the assumption cannot reasonably be assumed to hold. For instance, in many settings, only one of the two required types of proxies may be available rendering point identification infeasible, even if the proxies are highly relevant. More broadly, point identification of causal effects using existing proximal causal inference methods crucially relies on identification of a confounding bridge function, and remains possible even if the latter is only set-identified \citep{zhang2023proximal, bennett2023minimax}. Notably, failure to identify such a confounding bridge function is generally inevitable without a completeness condition; in which case partial identification of causal effects may be the most one can hope for.

In this paper, we propose partial identification methods for causal parameters using proxy variables in settings where completeness assumptions do not necessarily hold. Therefore, our methods obviate the need for identification of bridge functions which are solutions to integral equations that involve unobserved variables. We provide methods in single proxy settings and demonstrate that certain conditional independence requirements in that setting can be relaxed in case that the researcher has access to a second proxy variable. For the causal parameter of interest, we initially focus on the average treatment effect (ATE) and the effect of the treatment on the treated (ETT) in Sections \ref{sec:singleATE} and \ref{sec:doubleATE}. Then we extend our results in Section \ref{sec:HidMed} to partial identification in related settings where identification hinges upon hidden mediators for which proxies are available, however such proxies fail to satisfy a completeness condition which would in principle guarantee point identification of the causal effect of interest; such settings include causal mediation with a mis-measured mediator and hidden front-door models. Our bounds are non-smooth functionals of the observed data distribution. In Section \ref{sec:est}, we construct smooth lower and upper bounds using LogSumExp approximations and derive their efficient influence functions for both conditional and marginal potential outcome means. The influence-function results support one-step estimators and analytic standard errors, while regular plug-in estimators may also be handled by the bootstrap.

\section{Partial Identification Using a Single Proxy}
\label{sec:singleATE}

Consider a setting with a binary treatment variable $A\in\{0,1\}$, an outcome variable $Y\in\mathcal{Y}\subseteq[0,+\infty)$,\footnote{We use capital calligraphic letters to denote the alphabet of the corresponding random variable.} and observed and unobserved confounders denoted by $X$ and $U$, respectively.  
We are interested in identifying the effect of the treatment on the treated (ETT), which is defined as
\[
\theta_{ETT}=\E[Y^{(A=1)}\mid A=1]-\E[Y^{(A=0)}\mid A=1],
\]
as well as the average treatment effect (ATE) of $A$ on $Y$ defined as
\[
\theta_{ATE}=\E[Y^{(A=1)}]-\E[Y^{(A=0)}],
\]
where $Y^{(A=a)}$ is the potential outcome variable representing the outcome had the treatment (possibly contrary to fact) been set to value $a$. Due to the presence of the latent confounder $U$ in the system, without any extra assumptions, ETT and ATE are not point identified. Therefore, we focus on partial identification of  conditional potential outcome mean parameters, $\E[Y^{(A=a)}\mid A=1-a]$, and marginal potential outcome mean, $\E[Y^{(A=a)}]$, for $a\in\{0,1\}$.

\subsection{Partial Identification Using an Outcome Confounding Proxy}
\label{sec:singleATE:outcome}

Let $W$ be a proxy of the unobserved confounder, which can be say an error-prone measurement of $U$, or a so-called negative control outcome \citep{lipsitch2010negative,shi2020selective}. We show that under some requirements on $W$, the proxy can be leveraged to obtain bounds on the parameters $\E[Y^{(A=a)}\mid A=1-a]$ and $\E[Y^{(A=a)}]$. Our requirements on $W$ are as follows.

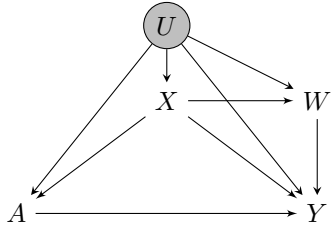
\begin{figure}[t!]
\centering
		\tikzstyle{block} = [draw, circle, inner sep=2.5pt, fill=lightgray]
		\tikzstyle{input} = [coordinate]
		\tikzstyle{output} = [coordinate]
        \begin{tikzpicture}
            \tikzset{edge/.style = {->,> = latex'}}
            \node[] (a) at  (-2,0) {$A$};
            \node[block] (u) at  (0,2.5) {$U$};
            \node[] (y) at  (2,0) {$Y$};
            \node[] (x) at  (0,1.5) {$X$};
            \node[] (w) at  (2,1.5) {$W$}; 
            \draw[-stealth] (u) to (a);
			\draw[-stealth] (u) to (y);
			\draw[-stealth][edge] (a) to (y);
			\draw[-stealth][edge] (x) to (a);			
            \draw[-stealth] (u) to (x);
            \draw[-stealth][edge] (x) to (y);
            \draw[-stealth][edge] (x) to (w);            
            \draw[-stealth] (u) to (w);  
            \draw[-stealth] (w) to (y);                                                          
        \end{tikzpicture}
        \caption{Example of a graphical model satisfying Assumption \ref{ass:indepM1}.}
        \label{fig:ATEsingle}
\end{figure}

\begin{assumption}
\label{ass:indepM1}
The proxy variable is independent of the treatment variable conditional on the confounders, i.e., $W\independent A\mid \{X,U\}$.	
\end{assumption}
We refer to a proxy variable $W$ satisfying Assumption \ref{ass:indepM1} as an outcome confounding proxy variable.

\begin{assumption}
\label{ass:hexistsM1}
There exists a non-negative bridge function $h$ such that almost surely
\[
\E[Y\mid A,X,U]=\E[h(W,A,X)\mid A,X,U].
\]
\end{assumption}
Assumption \ref{ass:indepM1} indicates how the proxy variable is related to other variables in the system. Mainly it states that as a proxy for $U$, $W$ would become irrelevant about the treatment process, and therefore not needed for confounding control had $U$ been observed. This is clearly a non-testable assumption which must therefore be grounded in prior expert knowledge. Figure \ref{fig:ATEsingle} demonstrates an example of a graphical model that satisfies Assumption \ref{ass:indepM1}.
 Assumption \ref{ass:hexistsM1} requires the existence of a so-called outcome confounding bridge function which depends on $W$ and whose conditional expectation matches the potential outcome mean, given $U$. The assumption formalizes the idea that $W$ is sufficiently rich so that there exists at least one transformation of the latter which can on average recover the potential outcome as a function of unmeasured confounders. The assumption was initially introduced by \cite{miao2018confounding} who give sufficient completeness and other regularity conditions for point identification of such a bridge function.  Note that we deliberately avoid such assumptions, and therefore cannot generally identify or estimate such a bridge function from observational data, because the integral equation defined in Assumption \ref{ass:hexistsM1} involves the unobserved confounder $U$. 
Besides Assumptions \ref{ass:indepM1} and \ref{ass:hexistsM1}, we also have the following requirements on the variables.
\begin{assumption}
\label{ass:usual}
~\\
\vspace{-7mm}
\begin{itemize}
\item For all values $a$, we have $Y^{(A=a)}\independent A\mid \{X,U\}$.	
\item For all values $a$, we have $p(A=a\mid X,U)>0$.
\item If $A=a$, we have $Y^{(A=a)}=Y$.
\end{itemize}

\end{assumption}
Part 1 of Assumption \ref{ass:usual} requires independence of the potential outcome variable and the treatment variable conditioned on both observed and unobserved confounders. Note that this is a much milder assumption compared to the standard conditional exhchangeability assumption as we also condition on the unobserved confounders. Parts 2 and 3 of Assumption \ref{ass:usual} are the positivity and consistency assumptions, which are standard assumptions in the field of causal identification.

 We have the following partial identification result for the conditional potential outcome mean.

\begin{theorem}
\label{thm:ETTsingleW}
Under Assumptions \ref{ass:indepM1}-\ref{ass:usual}, the parameter $\E[Y^{(A=a)}\mid A=1-a]$ can be bounded as follows.
\begin{align*}
&\max\Big\{\inf\mathcal{Y}~,~\E\Big[\min_w\frac{p(w\mid A=1-a,X)}{p(w\mid A=a,X)}\E[Y\mid A=a,X]\Big| A=1-a\Big]\Big\}\\
&\le\E[Y^{(A=a)}\mid A=1-a]\le\\
&\min\Big\{\sup\mathcal{Y}~,~\E\Big[\max_w\frac{p(w\mid A=1-a,X)}{p(w\mid A=a,X)}\E[Y\mid A=a,X]\Big| A=1-a\Big]\Big\}.
\end{align*}
\end{theorem}
 
 We have the following corollary for partial identification of marginal potential outcome mean.
  
\begin{corollary}
\label{cor:ETTtoATE}
Under Assumptions \ref{ass:indepM1}-\ref{ass:usual}, the parameter $\E[Y^{(A=a)}]$ can be bounded as follows.
\begin{align*}
&\max\Big\{\inf\mathcal{Y}\times p(A=1-a)+\E[Y\mid A=a]p(A=a),\E\Big[\frac{I(A=a)}{\max_w p(a\mid w,X)}Y\Big]\Big\}\\
&\le\E[Y^{(A=a)}]\le\\
&\min\Big\{\sup\mathcal{Y}\times p(A=1-a)+\E[Y\mid A=a]p(A=a),\E\Big[\frac{I(A=a)}{\min_w p(a\mid w,X)}Y\Big]\Big\},
\end{align*}
where $I(\cdot)$ is the indicator function.
\end{corollary}
Several remarks are in order regarding the proposed bounds.
\begin{remark}
Note that if there are no unobserved confounders in the system, $A$ and $W$ are independent conditional on $X$ and hence, the upper and lower bound in Theorem \ref{thm:ETTsingleW} and Corollary \ref{cor:ETTtoATE} will match. Especially, those in Corollary \ref{cor:ETTtoATE} will be equal to the standard inverse probability weighting (IPW) identification formula in the conditional exchangeability framework. Furthermore, the bounds can be written as 
\begin{align*}
&\max\Big\{\inf\mathcal{Y}\times p(A=1-a)+\E[Y\mid A=a]p(A=a),	\E\Big[\frac{p(a\mid X)}{\max_w p(a\mid w,X)}\E[Y\mid A=a,X]\Big]\Big\}\\
&\le\E[Y^{(A=a)}]\le\\
&\min\Big\{\sup\mathcal{Y}\times p(A=1-a)+\E[Y\mid A=a]p(A=a),\E\Big[\frac{p(a\mid X)}{\min_w p(a\mid w,X)}\E[Y\mid A=a,X]\Big]\Big\},
\end{align*}
which demonstrates the connection to the g-formula in the conditional exchangeability framework \citep{hernan2020causal}.
Formally, when $A$ and $W$ are independent conditioned on $X$, we have $\E\Big[\frac{p(a\mid X)}{\max_w p(a\mid w,X)}\E[Y\mid A=a,X]\Big]=\E\big[\E[Y\mid A=a,X]\big]$, which is the $g$-formula, and for the lower bound, we have
\begin{align*}
&\inf\mathcal{Y}\times p(A=1-a)+\E[Y\mid A=a]p(A=a)\\
&=\sum_x\E[\inf\mathcal{Y}\mid A=a,X=x]\times p(A=1-a,X=x)+\sum_x\E[Y\mid A=a,X=x]p(A=a,X=x)\\
&\le\sum_x\E[Y\mid A=a,X=x]\times p(A=1-a,X=x)+\sum_x\E[Y\mid A=a,X=x]p(A=a,X=x)\\
&=\E\big[\E[Y\mid A=a,X]\big].
\end{align*}
That is, the lower bound will be the g-formula. Similar argument holds for the upper bound.

\end{remark}

\begin{remark}
In the setting considered in this work we assumed $\mathcal{Y}\subseteq[0,+\infty)$. This assumption is merely for the sake of clarity of the presentation and can be relaxed to assuming the outcome is bounded from below. If $\inf\mathcal{Y}<0$, without loss of generality, for the bridge function $h$, we can have $\inf\mathcal{Y}\le h'_{m}\le h(w,a,x)$, for all $w,a,x$. Let $h_m=|\min\{h'_m,0\}|$. Then the bounding step of the proof of Theorem \ref{thm:ETTsingleW} should be changed to
\begin{align*}
&\sum_{x}\min_w\frac{p(w\mid A=1-a,x)}{p(w\mid A=a,x)}\sum_{w}h(w,a,x)p(w\mid A=a,x)p(x\mid A=1-a)\\
&\qquad+h_m\Big\{\E\Big[\min_w\frac{p(w\mid A=1-a,X)}{p(w\mid A=a,X)}\Big| A=1-a\Big]-1 \Big\}\\
&\le\E[Y^{(A=a)}\mid A=1-a]\le\\
&\sum_{x}\max_w\frac{p(w\mid A=1-a,x)}{p(w\mid A=a,x)}\sum_{w}h(w,a,x)p(w\mid A=a,x)p(x\mid A=1-a)\\
&\qquad+h_m\Big\{\E\Big[\max_w\frac{p(w\mid A=1-a,X)}{p(w\mid A=a,X)}\Big| A=1-a\Big]-1 \Big\}.
\end{align*}
\end{remark}

\begin{remark}
If we are only interested in partially identifying ETT, then we only need bounds for the parameter $\E[Y^{(A=0)}\mid A=1]$. In this case, Assumption \ref{ass:hexistsM1} can be weakened to requiring the existence of a bridge function $h$ such that almost surely
\[
\E[Y\mid A=0,X,U]=\E[h(W,X)\mid X,U].
\]
\end{remark}

\begin{remark}
In case that we have access to several proxy variables $\{W_1,...,W_K\}$ which satisfy Assumptions \ref{ass:indepM1} and \ref{ass:hexistsM1}, a method for improving the bounds would be to construct bounds based on each proxy variable and then choose the upper bound to be the minimum of the upper bounds and choose the lower bound to be the maximum of the lower bounds.	
\end{remark}

\subsection{Partial Identification Using a Treatment Confounding Proxy}
\label{sec:singleATE:treatment}

The proxy variable $W$ is potentially directly associated with the outcome variable but is required to be conditionally independent from the treatment variable. In this subsection, we establish an alternative partial identification method based on a proxy variable $Z$ of the unobserved confounder, which can be directly associated with the treatment variable, but is required to be conditionally independent of the outcome variable. Formally, we have the following requirements on the proxy variable $Z$.

\begin{assumption}
\label{ass:indepM2}
The proxy variable $Z$ is independent of the outcome variable conditional on the confounders and the treatment variable, i.e.,  $Z\independent Y\mid A,X,U$.	
\end{assumption}
We refer to a proxy variable satisfying Assumption \ref{ass:indepM2} as a treatment confounding proxy variable.
Figure \ref{fig:ATEsingle2} demonstrates an example of a graphical model that satisfies Assumption \ref{ass:indepM2}.
\begin{figure}[t!]
\centering
		\tikzstyle{block} = [draw, circle, inner sep=2.5pt, fill=lightgray]
		\tikzstyle{input} = [coordinate]
		\tikzstyle{output} = [coordinate]
        \begin{tikzpicture}
            \tikzset{edge/.style = {->,> = latex'}}
            \node[] (a) at  (-2,0) {$A$};
            \node[block] (u) at  (0,2.5) {$U$};
            \node[] (y) at  (2,0) {$Y$};
            \node[] (x) at  (0,1.5) {$X$};
            \node[] (z) at  (-2,1.5) {$Z$}; 
            \draw[-stealth] (u) to (a);
			\draw[-stealth] (u) to (y);
			\draw[-stealth][edge] (a) to (y);
			\draw[-stealth][edge] (x) to (a);			
            \draw[-stealth] (u) to (x);
            \draw[-stealth][edge] (x) to (y);
            \draw[-stealth][edge] (x) to (z);            
            \draw[-stealth] (u) to (z);  
            \draw[-stealth] (z) to (a);                                                          
        \end{tikzpicture}
        \caption{Example of a graphical model satisfying Assumption \ref{ass:indepM2}.}
        \label{fig:ATEsingle2}
\end{figure}
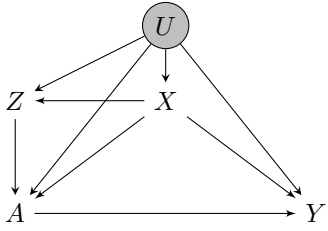

\begin{assumption}
\label{ass:qexistsM2}
There exists a non-negative bridge function $q$ such that almost surely
\[
\E[q(Z,A,X)\mid A,X,U]=\frac{p(U\mid 1-A,X)}{p(U\mid A,X)}.
\]
\end{assumption}
Existence of similar treatment confounding bridge function was first discussed by \cite{deaner2018proxy} and \cite{cui2023semiparametric} who also imposed additional conditions, including completeness, such that the treatment bridge function and causal effects are both uniquely nonparametrically identified. As in the previous section, we do not make such additional assumptions. 

We have the following partial identification result for the conditional potential outcome mean.

\begin{theorem}
\label{thm:ETTsingleZ}
Under Assumptions \ref{ass:usual}-\ref{ass:qexistsM2}, the parameter $\E[Y^{(A=a)}\mid A=1-a]$ can be bounded as follows.
\begin{align*}
&\E\big[\min_z\E[Y\mid z,X,A=a]\big| A=1-a\big]\\
&\le\E\big[Y^{(A=a)}\mid A=1-a]\le\\
&\E\big[\max_z\E[Y\mid z,X,A=a]\big| A=1-a\big].
\end{align*}
\end{theorem}

 We have the following corollary for partial identification of marginal potential outcome mean.
  
\begin{corollary}
\label{cor:ETTtoATEZ}
Under Assumptions \ref{ass:usual}-\ref{ass:qexistsM2}, the parameter $\E[Y^{(A=a)}]$ can be bounded as follows.
\begin{align*}
&\E\big[\min_z\E[Y\mid z,X,A=a]\big| A=1-a\big]p(A=1-a)+\E[Y\mid A=a]p(A=a)\\
&\le\E[Y^{(A=a)}]\le\\
&\E\big[\max_z\E[Y\mid z,X,A=a]\big| A=1-a\big]p(A=1-a)+\E[Y\mid A=a]p(A=a).
\end{align*}
\end{corollary}
Several remarks are in order regarding the proposed bounds.
\begin{remark}
Note that in Theorem \ref{thm:ETTsingleZ}, the lower bound will not be smaller than $\inf\mathcal{Y}$ and hence unlike Theorem \ref{thm:ETTsingleW}, we do not need to consider $\inf\mathcal{Y}$. Similarly, the upper bound will not be larger than $\sup\mathcal{Y}$.
\end{remark}

\begin{remark}
Note that if there are no unobserved confounders in the system, $Y$ and $Z$ are independent conditional on $A$ and $X$, providing a clear connection of the bounds in Theorem \ref{thm:ETTsingleZ} and Corollary \ref{cor:ETTtoATEZ} to the IPW and g-formula identification methods in the conditional exchangeability framework given $X$ only.
\end{remark}

\begin{remark}
If we are only interested in partially identifying ETT, then we only need bounds for the parameter $\E[Y^{(A=0)}\mid A=1]$. In this case, Assumption \ref{ass:qexistsM2} can be weakened to requiring the existence of a bridge function $q$ such that almost surely
\[
\E[q(Z,X)\mid A=0,X,U]=\frac{p(U\mid A=1,X)}{p(U\mid A=0,X)}.
\]
\end{remark}

\begin{remark}
Similar to the case of outcome confounding proxies,
in case that we have access to several proxy variables $\{Z_1,...,Z_K\}$ which satisfy Assumptions \ref{ass:indepM2} and \ref{ass:qexistsM2}, a method for improving the bounds would be to construct bounds based on each proxy variable and then choose the upper bound to be the minimum of the upper bounds and choose the lower bound to be the maximum of the lower bounds.	
\end{remark}

\section{Partial Identification Using Two Independent Invalid Proxies}
\label{sec:doubleATE}
In this section, we show that having access to two conditionally independent but invalid proxy variables of the unobserved confounder, one can relax the requirements of conditional independence of proxy variables from the treatment and the outcome variables while still providing valid bounds for the treatment effect. The proxy variables in our setting do not necessarily satisfy the strong exclusion restrictions of the original proximal causal inference. Therefore, the point identification approach in the original framework cannot be used here. We require the existence of two proxy variables $W$ and $Z$ that satisfy the following.

\begin{figure}[t!]
\centering
		\tikzstyle{block} = [draw, circle, inner sep=2.5pt, fill=lightgray]
		\tikzstyle{input} = [coordinate]
		\tikzstyle{output} = [coordinate]
        \begin{tikzpicture}
            \tikzset{edge/.style = {->,> = latex'}}
            \node[] (a) at  (-2,0) {$A$};
            \node[block] (u) at  (0,2.5) {$U$};
            \node[] (y) at  (2,0) {$Y$};
            \node[] (x) at  (0,1.5) {$X$};
            \node[] (w) at  (2,1.5) {$W$};
            \node[] (z) at  (-2,1.5) {$Z$};              
            \draw[-stealth] (u) to (a);
			\draw[-stealth] (u) to (y);
			\draw[-stealth][edge] (a) to (y);
			\draw[-stealth][edge] (x) to (a);			
            \draw[-stealth] (u) to (x);
            \draw[-stealth][edge] (x) to (y);
            \draw[-stealth][edge] (x) to (w);            
            \draw[-stealth] (u) to (w);  
            \draw[-stealth] (w) to (y);
            \draw[-stealth][edge] (x) to (z);            
            \draw[-stealth] (u) to (z);  
            \draw[-stealth] (z) to (a);
			\draw[-stealth,red] (a) to (w);            
            \draw[-stealth,red] (z) to (y);
        \end{tikzpicture}
                \caption{Example of a graphical model consistent with the two proxy setting.}
        \label{fig:ATEdouble3}
\end{figure}
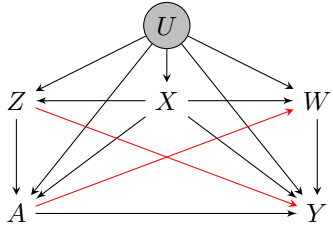

\begin{assumption}
\label{ass:M5}
~\\
\vspace{-7mm}
\begin{itemize}
\item The proxy variables are independent conditional on the treatment variable and all confounders, i.e., $W\independent Z\mid \{A,X,U\}$.
\item The invalid proxy variables satisfy Assumptions \ref{ass:hexistsM1} and \ref{ass:qexistsM2}. 	
\end{itemize}
\end{assumption}
Figure \ref{fig:ATEdouble3} demonstrates an example of a graphical model consistent with our two-proxy setting. Note that neither proxy variable $W$ technically needs to be an outcome confounding proxy variable (defined in Assumption \ref{ass:indepM1}) nor proxy variable $Z$ needs to be a treatment confounding proxy variable (defined in Assumption \ref{ass:indepM2}). Specifically, there can exist a direct causal link between $A$ and $W$ and one between $Z$ and $Y$ (shown in the figure by the red edges).

We have the following partial identification result for the conditional potential outcome mean.

\begin{theorem}
\label{thm:ETTdouble}
Under Assumptions \ref{ass:usual} and \ref{ass:M5}, the parameter $\E[Y^{(A=a)}\mid A=1-a]$ can be bounded as follows.
\begin{align*}
&\max\Big\{\inf\mathcal{Y}~,~\E\Big[\min_{w,z}\frac{p(w,z\mid a,X)}{p(w\mid a,X)p(z\mid a,X)}\E[Y\mid A=a,X]\Big|A=1-a\Big]\Big\}\\
&\le\E[Y^{(A=a)}\mid A=1-a]\le\\
&\min\Big\{\sup\mathcal{Y}~,~\E\Big[\max_{w,z}\frac{p(w,z\mid a,X)}{p(w\mid a,X)p(z\mid a,X)}\E[Y\mid A=a,X]\Big|A=1-a\Big]\Big\}.
\end{align*}
\end{theorem}

Note that in the absence of the unobserved confounder $U$, $\min_{w,z}p(w,z\mid a,x)=\max_{w,z}p(w,z\mid a,x)=p(w\mid a,x)p(z\mid a,x)$. Hence, again the upper and lower bounds match and will be equal to the g-formula and IPW formula in the conditional exchangeability framework given $X$.

\section{Partial Identification in the Presence of Unobserved Mediators}
\label{sec:HidMed}

In this Section we show that our partial identification ideas can also be used for dealing with unobserved mediators in the system. The results in this section are the counterparts of the results of \cite{ghassami2021hidmed} where the authors obtained point identification results under completeness assumptions. We focus on two settings with hidden mediators: Mediation analysis and front-door model. We only show the extension of the method presented in Subsection \ref{sec:singleATE:outcome} and do not repeat  extensions for the methods of Sections \ref{sec:singleATE:treatment} and \ref{sec:doubleATE}, although they are easily deduced from our exposition.

Let $M$ be the mediator variable in the system which is unobserved. We denote the potential outcome variable of $Y$, had the treatment and mediator variables been set to value $A=a$ and $M=m$ (possibly contrary to the fact) by $Y^{(a,m)}$. 
Similarly, we define $M^{(a)}$ as the potential outcome variable of $M$ had the treatment variables been set to value $A=a$.
Based on variables $Y^{(a,m)}$ and $M^{(a)}$, we define $Y^{(a)}=Y^{(a,M^{(a)})}$, and $Y^{(m)}=Y^{(A,m)}$.
We posit the following standard assumptions on the model.

\begin{assumption}
\label{assumption:usualmed}
~\\
\vspace{-7mm}
\begin{itemize}
\item For all $m$, $a$, and $x$, we have $p(m\mid a,x)>0$, and  $p(a\mid x)>0$.
\item $M^{(a)}=M$ if $A=a$. $Y^{(a,m)}=Y$ if $A=a$ and $M=m$.
\end{itemize}
\end{assumption}

\subsection{Partial Identification of NIE and NDE}

We first focus on mediation analysis. The goal is to identify the direct and indirect parts of the ATE, i.e., the part mediated through the mediator variable and the rest of the ATE. This can be done noting the following presentation of the ATE.
\begin{align*}
\E[Y^{(1)}-Y^{(0)}]
&=\E[Y^{(1,M^{(1)})}-Y^{(0,M^{(0)})}]\\
&=\E[Y^{(1,M^{(1)})}-Y^{(1,M^{(0)})}]
+\E[Y^{(1,M^{(0)})}-Y^{(0,M^{(0)})}].
\end{align*}
The first and the second terms in the last expression are called the total indirect effect and the pure direct effect, respectively by \cite{robins1992identifiability}, and are called the natural indirect effect (NIE) and the natural direct effect (NDE) of the treatment on the outcome, respectively by \cite{pearl2001direct}. We follow the latter terminology. NIE can be interpreted as the potential outcome mean if the treatment is fixed at $A=1$ from the point of view of the outcome variable but changes from $A=1$ to $A=0$ from the point of view of the mediator variable. Similarly, NDE can be interpreted as the potential outcome mean if the treatment is changed via intervention from $A=1$ to $A=0$ from the point of view of the outcome variable but is fixed at $A=0$ from the point of view of the mediator variable. We investigate identification of NIE and NDE in case the underlying mediator is hidden, however under the assumption that the model does not contain any hidden confounder. This can be formalized as follows \citep{imai2010general}. 

\begin{assumption}
\label{assumption:medclass}
For any two values of the treatment $a$ and $a'$, and value of the mediator $m$, we have
$(i)$ $Y^{(a,m)}\independent A\mid X$,
$(ii)$ $M^{(a)}\independent A\mid X$,
and $(iii)$ $Y^{(a,m)}\independent M^{(a')}\mid X$.
\end{assumption}
Due to Assumption \ref{assumption:medclass}, the parameters $\E[Y^{(1,M^{(1)})}]$ and $\E[Y^{(0,M^{(0)})}]$ are point identified. Hence, we focus on the parameter $\E[Y^{(1,M^{(0)})}]$. The identification formula for the this functional, known as the mediation formula in the literature \citep{pearl2001direct,imai2010general} requires observing the mediator variable. Here we show that in the setting with unobserved mediator, this parameter can be partially identified provided that we have access to a proxy variable $W$ of the hidden mediator. We have the following requirements on the proxy variable.

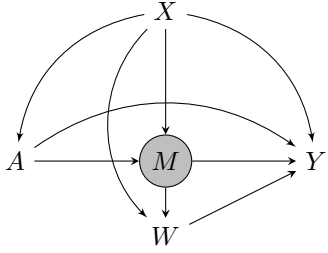
\begin{figure}[t!]
\centering
		\tikzstyle{block} = [draw, circle, inner sep=2.5pt, fill=lightgray]
		\tikzstyle{input} = [coordinate]
		\tikzstyle{output} = [coordinate]
        \begin{tikzpicture}
            \tikzset{edge/.style = {->,> = latex'}}
            \node[] (a) at  (-2,0) {$A$};
            \node[block] (m) at  (0,0) {$M$};
            \node[] (y) at  (2,0) {$Y$};
            \node[] (x) at  (0,2) {$X$};
            \node[] (w) at  (0,-1) {$W$}; 
            \draw[-stealth] (a) to (m);
			\draw[-stealth] (m) to (y);
			\draw[-stealth][edge, bend left=35] (a) to (y);
			\draw[-stealth][edge, bend left=-35] (x) to (a);
            \draw[-stealth] (x) to (m);
            \draw[-stealth][edge, bend left=35] (x) to (y);
            \draw[-stealth][edge, bend left=-45] (x) to (w);            
            \draw[-stealth] (m) to (w);  
            \draw[-stealth] (w) to (y);                                                          
        \end{tikzpicture}
                \caption{Example of a graphical model that satisfies Assumptions \ref{assumption:medclass} and \ref{ass:indepMNIE}.}
        \label{fig:mediation}        
\end{figure}

\begin{assumption}
\label{ass:indepMNIE}
The proxy variable is independent of the treatment variable conditional on the mediator variable and the observed confounder variable, i.e., $W\independent A\mid \{X,M\}$.
\end{assumption}
Figure \ref{fig:mediation} demonstrates an example of a graphical model that satisfies Assumptions \ref{assumption:medclass} and \ref{ass:indepMNIE}.

\begin{assumption}
\label{ass:hexistsMNIE}
There exists a non-negative bridge function $h$ such that almost surely
\[
\E[Y\mid A=1,X,M]=\E[h(W,X)\mid X,M].
\]
\end{assumption}

We have the following partial identification result for the parameter $\E[Y^{(1,M^{(0)})}]$.

\begin{theorem}
\label{thm:Mediation}	
Under Assumptions \ref{assumption:usualmed}-\ref{ass:hexistsMNIE}, the parameter $\E[Y^{(1,M^{(0)})}]$ can be bounded as follows.
\begin{align*}
&\max\Big\{\inf\mathcal{Y}~,~\E\Big[\min_w\frac{p(w\mid A=0,X)}{p(w\mid A=1,X)}\E[Y\mid A=1,X]\Big]\Big\}\\
&\le\E[Y^{(1,M^{(0)})}]\le\\
&\min\Big\{\sup\mathcal{Y}~,~\E\Big[\max_w\frac{p(w\mid A=0,X)}{p(w\mid A=1,X)}\E[Y\mid A=1,X]\Big]\Big\}.
\end{align*}
\end{theorem}

\subsection{Partial Identification of ATE in the Front-Door Model}

Front-door model is one of the earliest models introduced in the literature of causal inference in which the causal effect of a treatment variable on an outcome variable is identified despite allowing for the treatment-outcome relation to have an unobserved confounder and without any assumptions on the form of the structural equations of the variables \citep{pearl2009causality}. The main assumption of this framework is that the causal effect of the treatment variable on the outcome variable is fully relayed by a mediator variable, and neither the treatment-mediator relation nor the mediator-outcome relation have an unobserved confounder. This assumption is formally stated as follows.\\

\begin{assumption}
\label{assumption:fd}
~\\
\vspace{-7mm}
\begin{itemize}
\item For any value $a$ of the treatment and value $m$ of the mediator, we have
$(i)$ $M^{(a)}\independent A\mid X$, and 
$(ii)$ $Y^{(m)}\independent M\mid A,X$.
\item For any value $a$ of the treatment and value $m$ of the mediator, we have $Y^{(a,m)}= Y^{(m)}$.
\end{itemize}
\end{assumption}
The identification formula for the front-door model requires observations of the mediator variable. Here we show that in the setting with unobserved mediator, the causal effect of the treatment variable on the outcome variable can be partially identified provided that we have access to a proxy variable $W$ of the hidden mediator which satisfies Assumptions \ref{ass:indepMNIE} and the following assumption pertaining to the existence of a bridge function. Figure \ref{fig:FD} demonstrates an example of a graphical model that satisfies Assumptions \ref{ass:indepMNIE} and \ref{assumption:fd}.

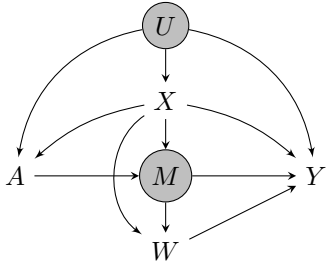
\begin{figure}[t!]
\centering
		\tikzstyle{block} = [draw, circle, inner sep=2.5pt, fill=lightgray]
		\tikzstyle{input} = [coordinate]
		\tikzstyle{output} = [coordinate]
        \begin{tikzpicture}
            \tikzset{edge/.style = {->,> = latex'}}
            \node[] (a) at  (-2,0) {$A$};
            \node[block] (m) at  (0,0) {$M$};
            \node[] (y) at  (2,0) {$Y$};
            \node[] (x) at  (0,1) {$X$};
            \node[block] (u) at  (0,2) {$U$};            
            \node[] (w) at  (0,-1) {$W$}; 
            \draw[-stealth] (a) to (m);
			\draw[-stealth] (m) to (y);
			\draw[-stealth] (u) to (x);			
			\draw[-stealth][edge, bend left=-15] (x) to (a);
			\draw[-stealth][edge, bend left=-35] (u) to (a);			
            \draw[-stealth] (x) to (m);
            \draw[-stealth][edge, bend left=15] (x) to (y);
            \draw[-stealth][edge, bend left=35] (u) to (y);            
            \draw[-stealth][edge, bend left=-55] (x) to (w);            
            \draw[-stealth] (m) to (w);  
            \draw[-stealth] (w) to (y);                                                          
        \end{tikzpicture}
                \caption{Example of a graphical model that satisfies Assumptions \ref{ass:indepMNIE} and \ref{assumption:fd}.}
        \label{fig:FD}        
\end{figure}

\begin{assumption}
\label{ass:hexistsFD}
There exists a non-negative bridge function $h$ such that almost surely
\[
\E[Y\mid A,X,M]=\E[h(W,A,X)\mid A,X,M].
\]	
\end{assumption}

We have the following partial identification result for the potential outcome mean.

\begin{theorem}
\label{thm:FrontDoor}	
Under Assumptions \ref{assumption:usualmed}, \ref{ass:indepMNIE}, \ref{assumption:fd}, and \ref{ass:hexistsFD}, the parameter $\E[Y^{(A=a)}]$ can be bounded as follows.
\begin{align*}
&\E[I(A=a)Y]+\max\Big\{\inf\mathcal{Y}\times p(A=1-a),\E\Big[I(A=1-a)\min_w\frac{p(w\mid a,X)}{p(w\mid 1-a,X)}Y\Big]\Big\}\\
&\le\E[Y^{(A=a)}]\le\\
&\E[I(A=a)Y]+\min\Big\{\sup\mathcal{Y}\times p(A=1-a),\E\Big[I(A=1-a)\max_w\frac{p(w\mid a,X)}{p(w\mid 1-a,X)}Y\Big]\Big\}.
\end{align*}

\end{theorem}

\section{Estimation and Inference}
\label{sec:est}

So far, we have presented nonparametric partial identification formulae for causal parameters in the presence of unobserved confounders or mediators. We now turn to estimation and inference. Throughout this section, $A$ is binary, $a\in\{0,1\}$ is fixed, and the proxy variables $W$ and $Z$ have fixed finite supports $\mathcal{W}$ and $\mathcal{Z}$. The covariates $X$ may be discrete or continuous. When the support endpoints enter a smooth functional below, $\inf\mathcal{Y}$ and $\sup\mathcal{Y}$ are treated as fixed, known, and finite constants. Let
\[
\mu_a(X)=\E[Y\mid A=a,X],\qquad
\mu_{a,z}(X)=\E[Y\mid A=a,Z=z,X],\qquad
p_a(w\mid X)=p(w\mid A=a,X).
\]
Thus, $p_{1-a}(w\mid X)=p(w\mid A=1-a,X)$. When all variables have finite support, the bounds may be estimated by empirical plug-in estimators. With continuous covariates, the same functionals may be evaluated using regression or machine learning estimates of the nuisance functions, subject to the regularity and rate conditions needed for the inferential procedure.

The bounds in Theorems \ref{thm:ETTsingleW} and \ref{thm:ETTsingleZ} are non-smooth functionals of the observed data law. We therefore use the LogSumExp operator to construct smooth lower and upper bounds. For real numbers $\{x_1,\ldots,x_n\}$ and $\gamma\ne 0$, define
\[
LSE(\{x_1,\ldots,x_n\};\gamma)
=\frac{1}{\gamma}\log\Big\{\sum_{i=1}^n\exp(\gamma x_i)\Big\}.
\]
For $\gamma>0$,
\[
\max_i x_i\le LSE(\{x_1,\ldots,x_n\};\gamma)
\le\max_i x_i+\frac{\log n}{\gamma},
\]
whereas, for $\gamma<0$,
\[
\min_i x_i+\frac{\log n}{\gamma}
\le LSE(\{x_1,\ldots,x_n\};\gamma)\le\min_i x_i.
\]
Consequently, the LSE operator approaches the maximum as $\gamma\to\infty$ and the minimum as $\gamma\to-\infty$. 
Applications of this approximation have appeared in optimization \citep{boyd2004convex}, machine learning \citep{murphy2012machine,goodfellow2016deep,calafiore2019log,calafiore2020universal}, and statistics \citep{wainwright2019high,chernozhukov2012central,tchetgen2017general,levis2023covariate}. Notably, \cite{tchetgen2017general} and \cite{levis2023covariate} employed this approximation technique to facilitate statistical inference pertaining to bounds within the instrumental variables model for missing data and causal effects, respectively.

The first result gives smooth bounds for the conditional potential outcome mean that enters the ETT.

\begin{corollary}
\label{cor:approx}
~
\vspace{-3mm}
\begin{itemize}
\item[\bf(a)] Let
\[
v_{w,X}=\frac{p_{1-a}(w\mid X)}{p_a(w\mid X)},\qquad
r_W(X;\gamma)=LSE(\{v_{w,X}\}_{w\in\mathcal{W}};\gamma),
\]
and define
\[
C(\gamma)=\E[r_W(X;\gamma)\mu_a(X)\mid A=1-a].
\]
For $\alpha>0$, set
\begin{align*}
\psi_{LW}(\alpha)
&=LSE(\{\inf\mathcal{Y},C(-\alpha)\};\alpha)-\frac{\log 2}{\alpha},\\
\psi_{UW}(\alpha)
&=LSE(\{\sup\mathcal{Y},C(\alpha)\};-\alpha)+\frac{\log 2}{\alpha}.
\end{align*}
Under the assumptions of Theorem \ref{thm:ETTsingleW},
\[
\psi_{LW}(\alpha)
\le \E[Y^{(A=a)}\mid A=1-a]
\le \psi_{UW}(\alpha).
\]

\item[\bf(b)] Let
\[
r_Z(X;\gamma)
=LSE(\{\mu_{a,z}(X)\}_{z\in\mathcal{Z}};\gamma),\qquad
\psi_Z(\gamma)=\E[r_Z(X;\gamma)\mid A=1-a],
\]
and, for $\alpha>0$, set $\psi_{LZ}(\alpha)=\psi_Z(-\alpha)$ and
$\psi_{UZ}(\alpha)=\psi_Z(\alpha)$. Under the assumptions of Theorem
\ref{thm:ETTsingleZ},
\[
\psi_{LZ}(\alpha)
\le \E[Y^{(A=a)}\mid A=1-a]
\le \psi_{UZ}(\alpha).
\]
\end{itemize}
\end{corollary}

Because the LSE operator is not invariant to multiplication by a positive scalar, smoothing $v_{w,X}$ before multiplication by $\mu_a(X)$ and smoothing the product directly yield different valid smooth bounds. The latter construction is recorded next.

\begin{corollary}
\label{cor:approxalt}
Let
\[
u_{w,X}=v_{w,X}\mu_a(X),\qquad
t_W(X;\gamma)=LSE(\{u_{w,X}\}_{w\in\mathcal{W}};\gamma),\qquad
D(\gamma)=\E[t_W(X;\gamma)\mid A=1-a].
\]
For $\alpha>0$, define
\begin{align*}
\widetilde{\psi}_{LW}(\alpha)
&=LSE(\{\inf\mathcal{Y},D(-\alpha)\};\alpha)-\frac{\log 2}{\alpha},\\
\widetilde{\psi}_{UW}(\alpha)
&=LSE(\{\sup\mathcal{Y},D(\alpha)\};-\alpha)+\frac{\log 2}{\alpha}.
\end{align*}
Under the assumptions of Theorem \ref{thm:ETTsingleW},
\[
\widetilde{\psi}_{LW}(\alpha)
\le \E[Y^{(A=a)}\mid A=1-a]
\le \widetilde{\psi}_{UW}(\alpha).
\]
\end{corollary}

The third result gives smooth bounds for the marginal potential outcome mean. These bounds may be combined across $a=0$ and $a=1$ to obtain smooth ATE bounds. To avoid a defect of the ordinary LSE soft minimum when it is used in a denominator, we use a positive soft extremum for the marginal $W$-based bounds.

\begin{corollary}
\label{cor:approxmarginal}
~
\vspace{-3mm}
\begin{itemize}
\item[\bf(a)] Let
\[
q_w(X)=p(A=a\mid W=w,X)
\]
and, for $\gamma\ne0$, define
\[
s_W(X;\gamma)
=\exp\!\left[LSE(\{\log q_w(X)\}_{w\in\mathcal{W}};\gamma)\right]
=\left\{\sum_{w\in\mathcal{W}}q_w(X)^\gamma\right\}^{1/\gamma}.
\]
Also let
\begin{align*}
C_1&=\inf\mathcal{Y}\,p(A=1-a)+\E[I(A=a)Y],\\
C_2&=\sup\mathcal{Y}\,p(A=1-a)+\E[I(A=a)Y],\\
C_3(\gamma)&=\E\left[\frac{I(A=a)Y}{s_W(X;\gamma)}\right],
\end{align*}
and define
\begin{align*}
\varphi_{LW}(\alpha)
&=LSE(\{C_1,C_3(\alpha)\};\alpha)-\frac{\log 2}{\alpha},\\
\varphi_{UW}(\alpha)
&=LSE(\{C_2,C_3(-\alpha)\};-\alpha)+\frac{\log 2}{\alpha}.
\end{align*}
Suppose $p(W=w\mid X)>0$ and $q_w(X)>0$ almost surely for every $w\in\mathcal{W}$. Under the assumptions of Corollary \ref{cor:ETTtoATE}, for every $\alpha>0$,
\[
\varphi_{LW}(\alpha)
\le \E[Y^{(A=a)}]
\le \varphi_{UW}(\alpha).
\]

\item[\bf(b)] Let $r_Z(X;\gamma)$ be as in Corollary \ref{cor:approx}, and define
\begin{align*}
\varphi_{LZ}(\alpha)
&=\E[I(A=1-a)r_Z(X;-\alpha)+I(A=a)Y],\\
\varphi_{UZ}(\alpha)
&=\E[I(A=1-a)r_Z(X;\alpha)+I(A=a)Y].
\end{align*}
Under the assumptions of Corollary \ref{cor:ETTtoATEZ}, for every $\alpha>0$,
\[
\varphi_{LZ}(\alpha)
\le \E[Y^{(A=a)}]
\le \varphi_{UZ}(\alpha).
\]
\end{itemize}
\end{corollary}

For regular plug-in nuisance estimators, the smoothness of the functionals in Corollaries \ref{cor:approx}--\ref{cor:approxmarginal} permits ordinary bootstrap inference under the usual regularity conditions. For brevity, consider the four endpoints in Corollary \ref{cor:approx}. Let $\widehat\psi_n(\alpha)$ denote a plug-in estimator, let $\widehat\psi_{n,1}^*(\alpha),\ldots,\widehat\psi_{n,B}^*(\alpha)$ denote its bootstrap replications, and let $\widehat\psi_{\beta}^*(\alpha)$ denote their empirical $\beta$-quantile. Basic-bootstrap $95\%$ interval estimators for the conditional potential outcome mean are
\[
C_{W_n}(\alpha)=\left(
2\widehat\psi_{LW,n}(\alpha)-\widehat\psi_{LW,1-0.05/2}^*(\alpha),
2\widehat\psi_{UW,n}(\alpha)-\widehat\psi_{UW,0.05/2}^*(\alpha)
\right)
\]
and
\[
C_{Z_n}(\alpha)=\left(
2\widehat\psi_{LZ,n}(\alpha)-\widehat\psi_{LZ,1-0.05/2}^*(\alpha),
2\widehat\psi_{UZ,n}(\alpha)-\widehat\psi_{UZ,0.05/2}^*(\alpha)
\right).
\]
Analogous formulas apply to Corollaries \ref{cor:approxalt} and \ref{cor:approxmarginal}.

\subsection{Efficient Influence Functions for the Smooth Bounds}

We next derive efficient influence functions (EIFs) for the smooth bound functionals under the nonparametric model for the observed data law \citep{newey1990semiparametric}, where $O=(Y,A,X,W,Z)$ denotes the observed data vector (with only the relevant proxy needed for each result). The following conditions are maintained for these results: $\alpha$ is fixed; $\mathcal{W}$ and $\mathcal{Z}$ are fixed finite supports; the fixed support endpoints used above are finite; $Y$ is square-integrable; and the treatment and proxy probabilities appearing in denominators are uniformly bounded away from zero on their relevant supports. In particular, this condition applies to $p(a\mid X)$, $p(1-a\mid X)$, $p_a(w\mid X)$, $p(w\mid X)$, $p(A=a,Z=z\mid X)$, and $q_w(X)$. These conditions are convenient sufficient conditions; more generally, it is enough that the displayed gradients are square-integrable under the corresponding inverse-weight moment conditions.

For the main conditional bounds, define
\[
\omega_w^W(X;\gamma)
=\frac{\exp\{\gamma v_{w,X}\}}
{\sum_{\widetilde w\in\mathcal{W}}\exp\{\gamma v_{\widetilde w,X}\}}
\]
and
\[
\omega_z^Z(X;\gamma)
=\frac{\exp\{\gamma\mu_{a,z}(X)\}}
{\sum_{\widetilde z\in\mathcal{Z}}\exp\{\gamma\mu_{a,\widetilde z}(X)\}}.
\]

\begin{theorem}
\label{thm:eifapprox}
Let
\begin{align*}
L_1(\alpha)
&=\frac{\exp\{\alpha C(-\alpha)\}}
{\exp\{\alpha\inf\mathcal{Y}\}+\exp\{\alpha C(-\alpha)\}},\\
L_2(\alpha)
&=\frac{\exp\{-\alpha C(\alpha)\}}
{\exp\{-\alpha\sup\mathcal{Y}\}+\exp\{-\alpha C(\alpha)\}}.
\end{align*}
Define
\begin{align*}
\eta_{C,\gamma}(O)
={}&\frac{I(A=1-a)}{p(A=1-a)}
\{r_W(X;\gamma)\mu_a(X)-C(\gamma)\}\\
&+\frac{I(A=a)}{p(a\mid X)}
\frac{p(1-a\mid X)}{p(A=1-a)}
r_W(X;\gamma)\{Y-\mu_a(X)\}\\
&+\sum_{w\in\mathcal{W}}
\frac{I(A=1-a)}{p(A=1-a)}
\omega_w^W(X;\gamma)\frac{\mu_a(X)}{p_a(w\mid X)}
\{I(W=w)-p_{1-a}(w\mid X)\}\\
&-\sum_{w\in\mathcal{W}}
\frac{I(A=a)}{p(a\mid X)}
\frac{p(1-a\mid X)}{p(A=1-a)}
\omega_w^W(X;\gamma)
\frac{p_{1-a}(w\mid X)\mu_a(X)}{p_a(w\mid X)^2}
\{I(W=w)-p_a(w\mid X)\}.
\end{align*}
Then the EIFs of $\psi_{LW}(\alpha)$ and $\psi_{UW}(\alpha)$ are, respectively,
\[
L_1(\alpha)\eta_{C,-\alpha}(O)
\qquad\text{and}\qquad
L_2(\alpha)\eta_{C,\alpha}(O).
\]

Define
\begin{align*}
\eta_{\psi_Z,\gamma}(O)
={}&\frac{I(A=1-a)}{p(A=1-a)}
\{r_Z(X;\gamma)-\psi_Z(\gamma)\}\\
&+\sum_{z\in\mathcal{Z}}
\frac{I(A=a,Z=z)}{p(A=a,Z=z\mid X)}
\frac{p(1-a\mid X)}{p(A=1-a)}
\omega_z^Z(X;\gamma)\{Y-\mu_{a,z}(X)\}.
\end{align*}
Then the EIFs of $\psi_{LZ}(\alpha)$ and $\psi_{UZ}(\alpha)$ are, respectively,
\[
\eta_{\psi_Z,-\alpha}(O)
\qquad\text{and}\qquad
\eta_{\psi_Z,\alpha}(O).
\]
\end{theorem}

For the alternative conditional smoothing, let
\[
\widetilde\omega_w^W(X;\gamma)
=\frac{\exp\{\gamma u_{w,X}\}}
{\sum_{\widetilde w\in\mathcal{W}}\exp\{\gamma u_{\widetilde w,X}\}}.
\]

\begin{theorem}
\label{thm:eifapproxalt}
Let
\begin{align*}
\widetilde L_1(\alpha)
&=\frac{\exp\{\alpha D(-\alpha)\}}
{\exp\{\alpha\inf\mathcal{Y}\}+\exp\{\alpha D(-\alpha)\}},\\
\widetilde L_2(\alpha)
&=\frac{\exp\{-\alpha D(\alpha)\}}
{\exp\{-\alpha\sup\mathcal{Y}\}+\exp\{-\alpha D(\alpha)\}}.
\end{align*}
Define
\begin{align*}
\eta_{D,\gamma}(O)
={}&\frac{I(A=1-a)}{p(A=1-a)}
\{t_W(X;\gamma)-D(\gamma)\}\\
&+\sum_{w\in\mathcal{W}}
\frac{I(A=a)}{p(a\mid X)}
\widetilde\omega_w^W(X;\gamma)
\frac{p(1-a\mid X)p_{1-a}(w\mid X)}
{p(A=1-a)p_a(w\mid X)}
\{Y-\mu_a(X)\}\\
&+\sum_{w\in\mathcal{W}}
\frac{I(A=1-a)}{p(A=1-a)}
\widetilde\omega_w^W(X;\gamma)
\frac{\mu_a(X)}{p_a(w\mid X)}
\{I(W=w)-p_{1-a}(w\mid X)\}\\
&-\sum_{w\in\mathcal{W}}
\frac{I(A=a)}{p(a\mid X)}
\widetilde\omega_w^W(X;\gamma)
\frac{p(1-a\mid X)p_{1-a}(w\mid X)\mu_a(X)}
{p(A=1-a)p_a(w\mid X)^2}
\{I(W=w)-p_a(w\mid X)\}.
\end{align*}
Then the EIFs of $\widetilde\psi_{LW}(\alpha)$ and
$\widetilde\psi_{UW}(\alpha)$ are, respectively,
\[
\widetilde L_1(\alpha)\eta_{D,-\alpha}(O)
\qquad\text{and}\qquad
\widetilde L_2(\alpha)\eta_{D,\alpha}(O).
\]
\end{theorem}

For the marginal $W$-based bounds, define
\[
\lambda_w(X;\gamma)
=\frac{q_w(X)^\gamma}
{\sum_{\widetilde w\in\mathcal{W}}q_{\widetilde w}(X)^\gamma},
\qquad
g_w(X;\gamma)
=\frac{\E[I(A=a)Y\mid X]\lambda_w(X;\gamma)}
{s_W(X;\gamma)q_w(X)}.
\]
Also let
\begin{align*}
K_1(\alpha)
&=\frac{\exp(\alpha C_1)}
{\exp(\alpha C_1)+\exp\{\alpha C_3(\alpha)\}},\\
K_2(\alpha)
&=\frac{\exp(-\alpha C_2)}
{\exp(-\alpha C_2)+\exp\{-\alpha C_3(-\alpha)\}},
\end{align*}
and define
\begin{align*}
\eta_1(O)&=\inf\mathcal{Y}\,I(A=1-a)+I(A=a)Y-C_1,\\
\eta_2(O)&=\sup\mathcal{Y}\,I(A=1-a)+I(A=a)Y-C_2,\\
\eta_{3,\gamma}(O)
&=\frac{I(A=a)Y}{s_W(X;\gamma)}-C_3(\gamma)\\
&\quad-\sum_{w\in\mathcal{W}}g_w(X;\gamma)
\frac{I(W=w)}{p(w\mid X)}
\{I(A=a)-q_w(X)\}.
\end{align*}

\begin{theorem}
\label{thm:eifmarginal}
The EIFs of $\varphi_{LW}(\alpha)$ and $\varphi_{UW}(\alpha)$ are, respectively,
\[
K_1(\alpha)\eta_1(O)+\{1-K_1(\alpha)\}\eta_{3,\alpha}(O)
\]
and
\[
K_2(\alpha)\eta_2(O)+\{1-K_2(\alpha)\}\eta_{3,-\alpha}(O).
\]

For the $Z$-based bounds, let
\[
M_Z(\gamma)=\E[I(A=1-a)r_Z(X;\gamma)]
\]
and define
\begin{align*}
\eta_{\varphi_Z,\gamma}(O)
={}&I(A=1-a)r_Z(X;\gamma)-M_Z(\gamma)\\
&+\sum_{z\in\mathcal{Z}}
\frac{I(A=a,Z=z)}{p(A=a,Z=z\mid X)}
p(1-a\mid X)\omega_z^Z(X;\gamma)
\{Y-\mu_{a,z}(X)\}\\
&+I(A=a)Y-\E[I(A=a)Y].
\end{align*}
Then the EIFs of $\varphi_{LZ}(\alpha)$ and $\varphi_{UZ}(\alpha)$ are, respectively,
\[
\eta_{\varphi_Z,-\alpha}(O)
\qquad\text{and}\qquad
\eta_{\varphi_Z,\alpha}(O).
\]
\end{theorem}

The displayed marginal EIFs immediately yield EIFs for the smooth ATE endpoints. For either proxy construction $Q\in\{W,Z\}$, write $\varphi_{LQ,a}(\alpha)$ and $\varphi_{UQ,a}(\alpha)$ for the endpoints evaluated at treatment level $a$, and let $\phi_{LQ,a}(O;\alpha)$ and $\phi_{UQ,a}(O;\alpha)$ denote their EIFs. Then
\[
L_{ATE,Q}(\alpha)=\varphi_{LQ,1}(\alpha)-\varphi_{UQ,0}(\alpha),
\qquad
U_{ATE,Q}(\alpha)=\varphi_{UQ,1}(\alpha)-\varphi_{LQ,0}(\alpha),
\]
have EIFs
\[
\phi_{LQ,1}(O;\alpha)-\phi_{UQ,0}(O;\alpha)
\qquad\text{and}\qquad
\phi_{UQ,1}(O;\alpha)-\phi_{LQ,0}(O;\alpha),
\]
respectively. Likewise, for any conditional smooth-bound pair $(\psi_{L,0}(\alpha),\psi_{U,0}(\alpha))$ evaluated at $a=0$, let $m_1=\E[Y\mid A=1]$. The corresponding ETT endpoints
\[
m_1-\psi_{U,0}(\alpha)
\qquad\text{and}\qquad
m_1-\psi_{L,0}(\alpha)
\]
have EIFs
\[
\frac{I(A=1)}{p(A=1)}(Y-m_1)-\phi_{\psi_{U,0}}(O;\alpha)
\qquad\text{and}\qquad
\frac{I(A=1)}{p(A=1)}(Y-m_1)-\phi_{\psi_{L,0}}(O;\alpha).
\]

For any one of the smooth endpoints $\Theta(P)$ above, its EIF $\phi_\Theta(O;\eta)$ may be used in a cross-fitted one-step estimator
\[
\widehat\Theta_{\mathrm{os}}
=\frac{1}{n}\sum_{i=1}^n
\left\{\Theta(\widehat P_{-k(i)})
+\phi_\Theta(O_i;\widehat\eta_{-k(i)})\right\}.
\]
If the nuisance estimators make the second-order remainder $o_p(n^{-1/2})$ and the estimated EIF converges in $L_2(P)$, then the estimator is asymptotically linear with influence function $\phi_\Theta$, and the empirical variance of the estimated EIF yields an analytic standard error. These additional rate conditions are required for one-step inference; pathwise differentiability alone is not sufficient.

For implementation, define the cross-fitted estimated influence values
\[
\widehat\phi_{\Theta,i}
=
\phi_\Theta(O_i;\widehat\eta_{-k(i)}),
\qquad
\overline{\widehat\phi}_{\Theta}
=
\frac{1}{n}\sum_{i=1}^n\widehat\phi_{\Theta,i},
\]
where all nuisance quantities entering the EIF are estimated using the training sample excluding fold $k(i)$. Estimate the asymptotic variance and standard error by
\[
\widehat\sigma_\Theta^2
=
\frac{1}{n-1}\sum_{i=1}^n
\left(
\widehat\phi_{\Theta,i}
-\overline{\widehat\phi}_{\Theta}
\right)^2,
\qquad
\widehat{\operatorname{se}}
\left(\widehat\Theta_{\mathrm{os}}\right)
=
\frac{\widehat\sigma_\Theta}{\sqrt n}.
\]
Thus, an endpoint-wise Wald $95\%$ confidence interval for $\Theta(P)$ is
\[
\left[
\widehat\Theta_{\mathrm{os}}
-z_{0.975}\frac{\widehat\sigma_\Theta}{\sqrt n},
\;
\widehat\Theta_{\mathrm{os}}
+z_{0.975}\frac{\widehat\sigma_\Theta}{\sqrt n}
\right],
\]
where $z_\beta$ denotes the $\beta$-quantile of the standard normal
distribution. More generally, for any smooth bound pair
$(\Theta_L(P),\Theta_U(P))$ and any causal parameter $\theta(P)$ satisfying
$\Theta_L(P)\leq\theta(P)\leq\Theta_U(P)$, the analytic EIF analogue of
the basic-bootstrap interval above is
\[
C_n^{\mathrm{EIF}}
=
\left[
\widehat\Theta_{L,\mathrm{os}}
-z_{0.975}\frac{\widehat\sigma_L}{\sqrt n},
\;
\widehat\Theta_{U,\mathrm{os}}
+z_{0.975}\frac{\widehat\sigma_U}{\sqrt n}
\right],
\]
where $\widehat\sigma_L^2$ and $\widehat\sigma_U^2$ are computed from the
corresponding estimated EIFs. Under the preceding one-step conditions
and consistent variance estimation, each of the two relevant one-sided
errors has asymptotic probability at most $0.025$. Hence, by Bonferroni's
inequality, $C_n^{\mathrm{EIF}}$ contains the entire population
smooth-bound interval with asymptotic probability at least $0.95$, and
therefore covers $\theta(P)$ with at least that probability. The same
construction applies to the conditional-mean, ATE, and ETT endpoint
pairs above, using their corresponding combined EIFs.

\section{Simulation Studies}
\label{sec:sims}

In this section, we provide simulation results to demonstrate the performance of our proposed estimators. We designed the data generating process for variables $(U,X,W,Z,A,Y)$ as follows. 
\begin{itemize}
\item $p(U,X)$: We chose $|\mathcal{U}|\times|\mathcal{X}|$ parameters for the joint distribution $p(U,X)$ uniformly from $Unif[0.1,1]$, and normalized them to sum up to one.
\item $p(W\mid U,X)$: For any $w,u,x$, 
\[
p(W=w\mid U=u,X=x)=\frac{\exp(\beta_{w,0}+\beta_{w,U}u+\beta_{w,X}x)}{\sum_{w\in\mathcal{W}}\exp(\beta_{w,0}+\beta_{w,U}u+\beta_{w,X}x)}.
\] 
\item $p(Z\mid U,X)$: For any $z,u,x$,
\[
p(Z=z\mid U=u,X=x)=\frac{\exp(\beta_{z,0}+\beta_{z,U}u+\beta_{z,X}x)}{\sum_{z\in\mathcal{Z}}\exp(\beta_{z,0}+\beta_{z,U}u+\beta_{z,X}x)}.
\]
\item $p(A\mid U,X,Z)$: For any $a,u,x,z$,
\[
p(A=a\mid U=u,X=x,Z=z)=\frac{\exp(\beta_{a,0}+\beta_{a,U}u+\beta_{a,X}x+\beta_{a,Z}z)}{\sum_{a\in\mathcal{A}}\exp(\beta_{a,0}+\beta_{a,U}u+\beta_{a,X}x+\beta_{a,Z}z)}.
\]
\item $p(Y\mid U,X,W,A)$: For any $y,u,x,w,a$, 
\[
p(Y=y\mid U=u,X=x,W=w,A=a)=\frac{\exp(\beta_{y,0}+\beta_{y,U}u+\beta_{y,X}x+\beta_{y,W}w+\beta_{y,A}a)}{\sum_{y\in\mathcal{Y}}\exp(\beta_{y,0}+\beta_{y,U}u+\beta_{y,X}x+\beta_{y,Z}z+\beta_{y,A}a)}.
\] 
\end{itemize}
In all the conditional distributions above, coefficients $\beta_{\cdot,\cdot}$ are chosen according to uniform distribution $Unif[-0.5,0.5]$. We require the coefficients connecting $U$ to $W$ and $Z$ to be non-zero to ensure that $W$ and $Z$ are relevant to the latent confounder $U$.

\subsection{Simulation Study 1}
We first considered the case that $|\mathcal{U}|=|\mathcal{W}|=|\mathcal{Z}|$. In this case, based on the results of \cite{miao2018identifying} and \cite{shi2020multiply}, the bridge functions $h$ and $q$ exist. Consequently, in conjunction with the properties of our data generating process, we can establish that the assumptions of both Theorems \ref{thm:ETTsingleW} and \ref{thm:ETTsingleZ} are satisfied, and hence our proposed method in Section \ref{sec:est} will yield valid 95\% confidence intervals for the conditional potential outcome mean.

\begin{figure}[t]
\centering
\includegraphics[scale=0.7]{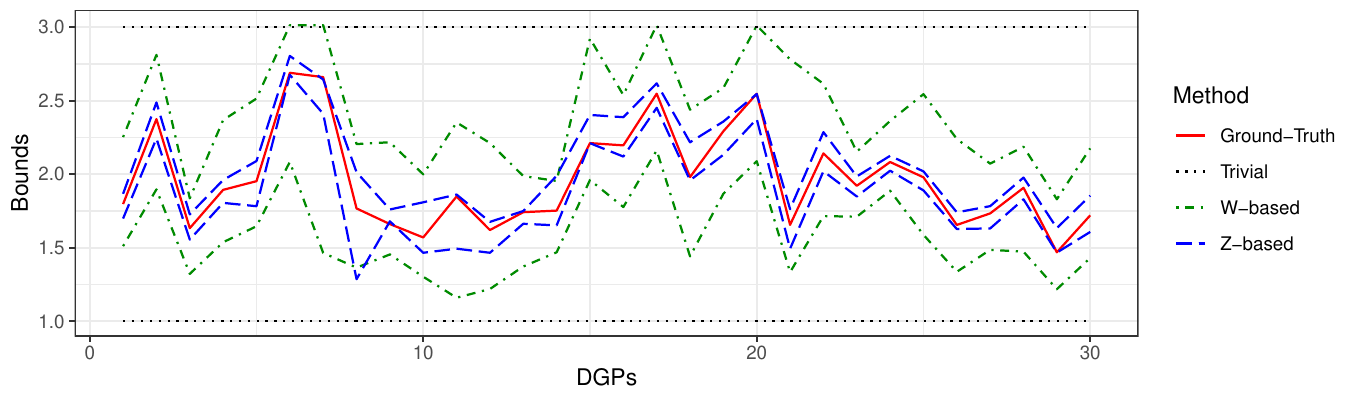}
\caption{Bounds for the conditional potential outcome mean for 30 random choices of the data generating process. Samples size from each data generating process is $n=5000$.}
\label{fig:many}
\end{figure}

\begin{table}[t]
\center
\def~{\hphantom{0}}
\caption{Average bound widths and average confidence interval widths for the $W$-based and $Z$-based methods.}{%
\begin{tabular}{c|cc|cc}
& $Z$-based method & & $W$-based method & \\
\hline
Sample Size & Avg. Bounds Width & Avg. CI Width & Avg. Bounds Width & Avg. CI Width\\
\hline
 3000 & 0.803 & 0.978 & 0.244 & 0.323\\
 4000 & 0.704 & 0.853 & 0.213 & 0.280\\
 5000 & 0.616 & 0.742 & 0.200 & 0.263\\
 6000 & 0.566 & 0.677 & 0.177 & 0.232\\
 7000 & 0.531 & 0.635 & 0.168 & 0.220\\
 8000 & 0.506 & 0.610 & 0.162 & 0.211\\
 9000 & 0.470 & 0.560 & 0.151 & 0.198
\end{tabular}}
\label{table:case1}
\end{table}

We considered $|\mathcal{U}|=|\mathcal{X}|=|\mathcal{W}|=|\mathcal{Z}|=4$, an outcome variable with $|\mathcal{Y}|=3$, and a binary treatment variable.
Table \ref{table:case1} presents the average bound width and the average confidence interval width for the conditional potential outcome mean for the $W$-based and $Z$-based methods, where the average is over 100 iterations (ground-truth value was $1.8$, trivial bounds for the parameter was $[1,3]$). The number of bootstrap replicates in each iteration is $B=500$ and the value of the hyper-parameter $\alpha$ is set to $50$. Notably, the $Z$-based bounds, i.e., the bounds based on Part (b) of Corollary \ref{cor:approx}, exhibit superior performance compared to the $W$-based bounds, i.e., the bounds based on Part (a) of Corollary \ref{cor:approx}. One may wonder whether the efficacy of the proposed method is contingent upon fortuitous realizations of the data generating process. Figure \ref{fig:many} serves to dispel this concern by illustrating that such dependence on chance is not observed. In that figure we observe the bounds for the conditional potential outcome mean for 30 random choices of the data generating process.

\subsection{Simulation Study 2}

\begin{table}[b]
\begin{minipage}{0.45\textwidth}
\centering
\def~{\hphantom{0}}
\caption{Average bound width for the $W$-based method.}{%
\begin{tabular}{c|ccccc}
  & & & $|\mathcal{W}|$ & &\\
 $|\mathcal{U}|$ & 3 & 4 & 5 & 6 & 7\\
\hline
 3 & 0.563 & 0.710 & 0.801 & 1.028 & 1.053\\
 4 & & 0.873 & 1.014 & 1.179 & 1.173\\
 5 & & & 1.148  & 1.213 & 1.330\\
 6 & & & & 1.378 & 1.396\\
 7 & & & & & 1.536
\end{tabular}}
\label{table:case2W}
\end{minipage}%
\begin{minipage}{0.08\textwidth}
~
\end{minipage}%
\begin{minipage}{0.45\textwidth}
\centering
\def~{\hphantom{0}}
\caption{Average bound width for the $Z$-based method.}{%
\begin{tabular}{c|ccccc}
  & & & $|\mathcal{Z}|$ & &\\
 $|\mathcal{U}|$ & 3 & 4 & 5 & 6 & 7\\
\hline
 3 & 0.161 & 0.168 & 0.194 & 0.281 & 0.285\\
 4 & & 0.225 & 0.253 & 0.356 & 0.383\\
 5 & & & 0.295 & 0.360 & 0.430\\
 6 & & & & 0.345 & 0.366\\
 7 & & & & & 0.532
\end{tabular}}
\label{table:case2Z}
\end{minipage}
\end{table}

Next we considered the case that $|\mathcal{W}|$ and $|\mathcal{Z}|$ are larger than or equal to $|\mathcal{U}|$. In this case, based on the results of \cite{shi2020multiply}, the bridge functions $h$ and $q$ exist and hence again the assumptions of Theorems 1 and 2 are satisfied. However, in this case, the bridge functions are not necessarily unique. Additionally, the proxy variables may incorporate information that is not directly pertinent to the latent confounder, resulting in wider bounds. This phenomenon is corroborated by the simulation results depicted in Tables \ref{table:case2W} and \ref{table:case2Z}. Each entry of the table is the width of the estimated bound averaged over 100 random data generating processes. The results are for sample size $n=10000$, $|\mathcal{X}|=5$, $|\mathcal{Y}|=3$, and a binary treatment variable. 

\begin{table}[t]
\caption{Coverage of the $W$-based and $Z$-based methods bounds for different values of $|\mathcal{W}|=|\mathcal{Z}|$.}
\centering
\begin{tabular}{c|cccc}
  & &  $|\mathcal{W}|=|\mathcal{Z}|$  \\
 & 3 & 4 & 5 & 6 \\
\hline
 $W$-based method & 98.4\% & 99.4\% & 100\% & 100\%\\
 $Z$-based method & 62.2\% & 76.2\% & 84.4\% & 92\%
\end{tabular}	
\label{table:sensitivity}
\end{table}

\subsection{Simulation Study 3}

Finally, we considered the case that $|\mathcal{W}|$ and $|\mathcal{Z}|$ are smaller than $|\mathcal{U}|$.
We considered $|\mathcal{U}|=7$, $|\mathcal{X}|=5$, $|\mathcal{Y}|=3$, $|\mathcal{W}|=|\mathcal{Z}|\in\{3,4,5,6\}$, and a binary treatment variable.  In this case, the bridge functions $h$ and $q$ do not exist and hence, the assumptions of Theorems 1 and 2 are violated. We investigated the coverage of our bounds to study the sensitivity of the approach to the existing slight violation of our assumption. To do so, we looked at 500 random data generating processes and investigated in what percentage of them the bounds obtained from a sample of size 10,000 contain the ground truth.  The results are shown in Table \ref{table:sensitivity}. The obtained coverages may demonstrate the robustness of our proposed methods to slight violation of the assumption of existence of bridge functions.

\section{Evaluation of the Effectiveness of Right Heart Catheterization}

In this section, we demonstrate the application of the proposed method to the Study to Understand Prognoses and Preferences for Outcomes and Risks of Treatments (SUPPORT) to evaluation of the effectiveness of right heart catheterization (RHC) in the intensive care unit of critically ill patients \citep{connors1996effectiveness}. The same dataset has been analyzed using proximal framework in \citep{tchetgen2020introduction,cui2023semiparametric} with parametric estimation on nuisance parameters, and in \citep{ghassami2021minimax} with non-parametric estimation of nuisance functions.

Data are available on 5735 individuals, 2184 treated and 3551 controls. In total, 3817 patients survived and 1918 died within 30 days.
The binary treatment variable $A$ is whether RHC is assigned, and the outcome variable $Y$ is the number of days between admission and death or censoring at day 30.
Based on background knowledge, we included the following five binary pre-treatment covariates to adjust for potential confounding: 
\begin{itemize}
	\item $X_1$: Indicator of age above 75; 
	\item $X_2$: Indicator of APACHE score above 40; 
	\item $X_3$: Indicator of estimate of probability of surviving two months above 0.5; 
	\item $X_4$: Indicator that patient has congestive heart failure or acute respiratory failure;
	\item $X_5$: Indicator of hematocrit above 30\%.
\end{itemize}
Based on the superior performance of the $Z$-based method observed in synthetic data evaluations, we only applied that method in our real data analysis. Hence, we only require proxy variable $Z$, for which, following \citep{tchetgen2020introduction,cui2023semiparametric}, we considered two options: 
(1) The status of PaO2/FI02 ratio (PaFI). Specifically, we used the binary variable $Z$ the indicator of PaO2/FI02 ratio above 150.
(2) The status of partial pressure of CO2 (PaCO2). Specifically, we used the binary variable $Z$ the indicator of PaCO2 above 37 mmHg.
The results are summarized in Table \ref{table:RHC}.
As it can be seen in Table \ref{table:RHC}, both choices of the proxy variable $Z$ lead to a negative causal effect of RHC on survival.
The results are consistent with the previous results in the literature on this dataset. Specifically, that of \cite{cui2023semiparametric} (ATE = $-1.66$, 95\% confidence interval = $(-2.50,-0.83)$) and \cite{ghassami2021minimax} (ATE = $-1.70$, 95\% confidence interval = $(-2.17, -1.22)$).

\begin{table}[t]
\centering
\caption{Causal effect estimates  and 95\% confidence intervals for two choices of the proxy variable $Z$.}
\begin{tabular}{c|cc}
 & ATE bounds & 95\% CIs \\
\hline
 PaFI as the proxy variable & $(-2.250,-0.093)$ & $(-2.738 ,  0.403)$\\
 PaCO2  as the proxy variable & $(-2.281, -0.038)$ & $(-2.721 , 0.466)$
\end{tabular}
\label{table:RHC}
\end{table}

\section{Conclusion}

For point identification of causal effects, proximal causal inference requires identification of certain nuisance functions called bridge functions using proxy variables that are sufficiently relevant to the unmeasured confounder, formalized as a completeness condition. However, completeness is not testable, and although a bridge function may exist, lack of completeness may severely limit prospects for identification of a bridge function and thus a causal effect; therefore, restricting the application of the framework. In this work, we proposed partial identification methods that do not require completeness and obviate the need for identification of a bridge function, i.e., we established that proxies can be leveraged to obtain bounds on the causal effect even if available information does not suffice to identify a bridge function. We further established analogous results for mediation analysis when the mediator is unobserved. Since our bounds are non-smooth functionals of the observed data distribution, for inference we proposed smooth lower and upper approximations and derived their efficient influence functions for both conditional and marginal potential outcome means. These results provide the basis for one-step estimators and analytic variance formulas, in addition to bootstrap inference for regular plug-in estimators. We provided detailed simulation results to demonstrate the performance of our proposed methods. Specifically, we observed that the $Z$-based method in many settings provide very informative bounds on the causal effect. We also demonstrated the application of our proposed method to the Study to Understand Prognoses and Preferences for Outcomes and Risks of Treatments (SUPPORT) to evaluation of the effectiveness of right heart catheterization in the intensive care unit of critically ill patients. Our results were consistent with the previous findings in the literature on this dataset.

\bibliography{Refs.bib}
\bibliographystyle{apalike}

\newpage
~\vspace{5mm}
\begin{center}
{\LARGE \bf Appendix}	
\end{center}
\vspace{10mm}
\appendix
\section*{Proofs}

\begin{proof}[Proof of Theorem \ref{thm:ETTsingleW}]
	
We note that
\begin{align*}
&\E[Y^{(A=a)}\mid A=1-a]\\
&=\E\Big[\E[Y^{(A=a)}\mid A=a,X,U]\Big| A=1-a\Big]\\
&\overset{(a)}{=}\E\Big[\E[Y\mid A=a,X,U]\Big| A=1-a\Big]\\
&\overset{(b)}{=}\E\Big[\E[h(W,a,X)\mid A=a,X,U]\Big| A=1-a\Big]\\
&\overset{(c)}{=}\E\Big[\E[h(W,a,X)\mid A=1-a,X,U]\Big| A=1-a\Big]\\
&=\E\Big[h(W,a,X)\Big| A=1-a\Big]\\
&=\sum_{w,x}h(w,a,x)\frac{p(w\mid A=1-a,x)}{p(w\mid A=a,x)}p(w\mid A=a,x)p(x\mid A=1-a),
\end{align*}
where 
$(a)$ is due to the consistency assumption, 
$(b)$ is due to Assumption \ref{ass:hexistsM1},
and $(c)$ is due to Assumption \ref{ass:indepM1}.
Therefore, 
\begin{align*}
&\sum_{x}\min_w\frac{p(w\mid A=1-a,x)}{p(w\mid A=a,x)}\sum_{w}h(w,a,x)p(w\mid A=a,x)p(x\mid A=1-a)\\
&\le\E[Y^{(A=a)}\mid A=1-a]\le\\
&\sum_{x}\max_w\frac{p(w\mid A=1-a,x)}{p(w\mid A=a,x)}\sum_{w}h(w,a,x)p(w\mid A=a,x)p(x\mid A=1-a).
\end{align*}
Note that Assumption \ref{ass:hexistsM1} implies that almost surely
\begin{equation*}
\E[Y\mid A=a,X]=\E[h(W,a,X)\mid A=a,X].	
\end{equation*}

Therefore,
\begin{align*}
&\sum_{x}\min_w\frac{p(w\mid A=1-a,x)}{p(w\mid A=a,x)}\sum_{y}yp(y\mid A=a,x)p(x\mid A=1-a)\\
&\le\E[Y^{(A=a)}\mid A=1-a]\le\\
&\sum_{x}\max_w\frac{p(w\mid A=1-a,x)}{p(w\mid A=a,x)}\sum_{y}yp(y\mid A=a,x)p(x\mid A=1-a).
\end{align*}
Combining the above bounds with the trivial bounds that $\inf\mathcal{Y}\le\E[Y^{(A=a)}\mid A=1-a]\le\sup\mathcal{Y}$ leads to the desired result.

\end{proof}

\begin{proof}[Proof of Corollary \ref{cor:ETTtoATE}]

We note that 
\[
\E[Y^{(a)}]=\E[Y^{(a)}\mid A=1-a]p(A=1-a)+\E[Y^{(a)}\mid A=a]p(A=a).
\]
Therefore, the following concludes the upper bound in the corollary.
\begin{align*}
&\sum_{x}\max_w\frac{p(w\mid A=1-a,x)}{p(w\mid A=a,x)}\sum_{y}yp(y\mid A=a,x)p(x\mid A=1-a)p(A=1-a)\\
&\qquad+\E[Y\mid A=a]p(A=a)\\
&=\sum_{x}\max_w\frac{p(A=1-a\mid w,x)}{p(A=a\mid w,x)}\frac{p(A=a\mid x)}{p(A=1-a\mid x)}\sum_{y}yp(y\mid A=a,x)p(x\mid A=1-a)p(A=1-a)\\
&\qquad+\E[Y\mid A=a]p(A=a)\\
&=\sum_{x}\max_w\frac{p(A=1-a\mid w,x)}{p(A=a\mid w,x)}\sum_{y}yp(y\mid A=a,x)p(x, A=a)+\E[Y\mid A=a]p(A=a)\\
&=\sum_{x}\Big\{1+\max_w\frac{p(A=1-a\mid w,x)}{p(A=a\mid w,x)}\Big\}\sum_{y}yp(y,A=a,x)\\
&=\sum_{x}\Big\{1+\max_w\frac{1}{p(A=a\mid w,x)}-1\Big\}\sum_{y}yp(y,A=a,x)\\
&=\sum_{x}\max_w\frac{1}{p(A=a\mid w,x)}\sum_{y}yp(y,A=a,x),\end{align*}

The lower bound can be proved similarly.

\end{proof}

\begin{proof}[Proof of Theorem \ref{thm:ETTsingleZ}]

We note that
\begin{align*}
&\E\big[Y^{(A=a)}\mid A=1-a]\\
&=\E\big[\E[Y^{(A=a)}\mid A=a,X,U]\mid A=1-a\big]\\
&\overset{(a)}{=}\E\big[\E[Y\mid A=a,X,U]\mid A=1-a\big]\\
&=\sum_{y,x,u}yp(y\mid A=a,x,u)\frac{p(u\mid x,A=1-a)}{p(u\mid x,A=a)}\frac{p(x\mid A=1-a)}{p(x\mid A=a)}p(x,u\mid A=a)\\
&\overset{(b)}{=}\sum_{z,y,x,u}yq(z,a,x)p(y\mid a,x,u)p(z\mid a,x,u)\frac{p(x\mid A=1-a)}{p(x\mid A=a)}p(x,u\mid A=a)\\
&\overset{(c)}{=}\sum_{z,y,x,u}yq(z,a,x)p(z,y,x,u\mid A=a)\frac{p(x\mid A=1-a)}{p(x\mid A=a)}\\
&=\sum_{z,y,x}yq(z,a,x)p(z,y\mid x,A=a)p(x\mid A=1-a)\\
&=\sum_{z,x}\E[Y\mid A=a,z,x]q(z,a,x)p(z\mid x,A=a)p(x\mid A=1-a),
\end{align*}
where 
$(a)$ is due to the consistency assumption,
$(b)$ is due to Assumption \ref{ass:qexistsM2},
and $(c)$ is due to Assumption \ref{ass:indepM2}. Therefore,
\begin{align*}
&\sum_{x}\min_z\E[Y\mid z,x,A=a]\sum_zq(z,a,x)p(z\mid A=a,x)p(x\mid A=1-a)\\
&\le\E\big[Y^{(A=a)}\mid A=1-a]\le\\
&\sum_{x}\max_z\E[Y\mid z,x,A=a]\sum_zq(z,a,x)p(z\mid A=a,x)p(x\mid A=1-a).\end{align*}

Note that Assumption \ref{ass:qexistsM2} implies that almost surely
\begin{equation*}
\begin{aligned}
\E[q(Z,&A,X)\mid A,X]\\
&=\sum_u\E[q(Z,A,X)\mid A,X,u]p(u\mid A,X)=\sum_up(u\mid 1-A,X)=1.
\end{aligned}
\end{equation*}
Therefore, 
\begin{align*}
&\sum_{x}\min_z\E[Y\mid z,x,A=a]p(x\mid A=1-a)\\
&\le\E\big[Y^{(A=a)}\mid A=1-a]\le\\
&\sum_{x}\max_z\E[Y\mid z,x,A=a]p(x\mid A=1-a).\end{align*}

\end{proof}

\begin{proof}[Proof of Theorem \ref{thm:ETTdouble}]

We note that
\begin{align*}
&\E[Y^{(A=a)}\mid A=1-a]\\
&\overset{(a)}{=}\sum_{y,x,u}yp(y\mid a,x,u)p(x,u\mid A=1-a)\\
&\overset{(b)}{=}\sum_{w,x,u}h(w,a,x)p(w\mid a,x,u)p(x,u\mid A=1-a)\\
&=\sum_{w,x,u}h(w,a,x)p(w\mid a,x,u)\frac{p(u\mid A=1-a,x)}{p(u\mid A=a,x)}\frac{p(x\mid A=1-a)}{p(x\mid A=a)}p(x,u\mid A=a)\\
&\overset{(c)}{=}\sum_{w,z,x,u}h(w,a,x)q(z,a,x)p(w\mid a,x,u)p(z\mid a,x,u)\frac{p(x\mid A=1-a)}{p(x\mid A=a)}p(x,u\mid A=a)\\
&\overset{(d)}{=}\sum_{w,z,x,u}h(w,a,x)q(z,a,x)p(w,z\mid a,x,u)\frac{p(x\mid A=1-a)}{p(x\mid A=a)}p(x,u\mid A=a)\\
&=\sum_{w,z,x}h(w,a,x)q(z,a,x)p(w,z\mid a,x)p(x\mid A=1-a),
\end{align*}
where 
$(a)$ is due to Assumption \ref{ass:usual},
$(b)$ is due to Assumption \ref{ass:hexistsM1},
$(c)$ is due to Assumption \ref{ass:qexistsM2}.
and $(d)$ is due to Assumption \ref{ass:M5}.
Therefore, 
\begin{align*}
&\sum_{x}\min_{w,z}\frac{p(w,z\mid a,x)}{p(w\mid a,x)p(z\mid a,x)}\sum_wh(w,a,x)p(w\mid a,x)\sum_zq(z,a,x)p(z\mid a,x)p(x\mid A=1-a)\\
&\le\E[Y^{(A=a)}\mid A=1-a]\le\\
&\sum_{x}\max_{w,z}\frac{p(w,z\mid a,x)}{p(w\mid a,x)p(z\mid a,x)}\sum_wh(w,a,x)p(w\mid a,x)\sum_zq(z,a,x)p(z\mid a,x)p(x\mid A=1-a).
\end{align*}
Note that
\[
\sum_wh(w,a,x)p(w\mid a,x)=\sum_yyp(y\mid a,x),
\]
and
\[
\sum_zq(z,a,x)p(z\mid a,x)=1.
\]
Therefore,
\begin{align*}
&\E\Big[\min_{w,z}\frac{p(w,z\mid a,X)}{p(w\mid a,X)p(z\mid a,X)}\E[Y\mid A=a,X]\Big|A=1-a\Big]\\
&\le\E[Y^{(A=a)}\mid A=1-a]\le\\
&\E\Big[\max_{w,z}\frac{p(w,z\mid a,X)}{p(w\mid a,X)p(z\mid a,X)}\E[Y\mid A=a,X]\Big|A=1-a\Big].
\end{align*}
Combining the above bounds with the trivial bounds that $\inf\mathcal{Y}\le\E[Y^{(A=a)}\mid A=1-a]\le\sup\mathcal{Y}$ leads to the desired result.
\end{proof}

\begin{proof}[Proof of Theorem \ref{thm:Mediation}]
We note that
\begin{align*}
&\E[Y^{(1,M^{(0)})}]\\
&=\E\Big[\E\big[\E[Y\mid A=1,X,M]\big|A=0,X\big]\Big]\\
&=\E\Big[\E\big[\E[h(W,X)\mid X,M]\big|A=0,X\big]\Big]\\
&=\E\Big[\E\big[\E[h(W,X)\mid A=0,X,M]\big|A=0,X\big]\Big]\\
&=\E\Big[\E[h(W,X)\mid A=0,X]\Big]\\
&=\E\Big[\E[\frac{p(W\mid A=0,X)}{p(W\mid A=1,X)}h(W,X)\mid A=1,X]\Big].
\end{align*}
Therefore,
\begin{align*}
&\E\Big[\min_w\frac{p(w\mid A=0,X)}{p(w\mid A=1,X)}\E[h(W,X)\mid A=1,X]\Big]\\
&\le\E[Y^{(1,M^{(0)})}]\le\\
&\E\Big[\max_w\frac{p(w\mid A=0,X)}{p(w\mid A=1,X)}\E[h(W,X)\mid A=1,X]\Big].
\end{align*}
Note that Assumption \ref{ass:hexistsMNIE} implies that almost surely
\[
\E[Y\mid A=1,X]=\E[h(W,X)\mid A=1,X].
\]
Therefore, 
\begin{align*}
&\E\Big[\min_w\frac{p(w\mid A=0,X)}{p(w\mid A=1,X)}\E[Y\mid A=1,X]\Big]\\
&\le\E[Y^{(1,M^{(0)})}]\le\\
&\E\Big[\max_w\frac{p(w\mid A=0,X)}{p(w\mid A=1,X)}\E[Y\mid A=1,X]\Big].
\end{align*}
Combining the above bounds with the trivial bounds that $\inf\mathcal{Y}\le\E[Y^{(1,M^{(0)})}]\le\sup\mathcal{Y}$ leads to the desired result.
	
\end{proof}

\begin{proof}[Proof of Theorem \ref{thm:FrontDoor}]
	
We note that
\begin{align*}
&\E[Y^{(A=a)}]\\
&=\E[I(A=a)Y]+\sum_{m,y,x}yp(y\mid A=1-a,m,x)p(m\mid A=a,x)p(A=1-a\mid x)p(x)\\
&=\E[I(A=a)Y]+\sum_{m,w,x}h(w,1-a,x)p(w\mid 1-a,m,x)p(m\mid a,x)p(1-a\mid x)p(x)\\
&=\E[I(A=a)Y]+\sum_{m,w,x}h(w,1-a,x)p(w\mid a,m,x)p(m\mid a,x)p(1-a\mid x)p(x)\\
&=\E[I(A=a)Y]+\sum_{w,x}h(w,1-a,x)p(w\mid a,x)p(1-a\mid x)p(x)\\
&=\E[I(A=a)Y]+\sum_{w,x}\frac{p(w\mid a,x)}{p(w\mid 1-a,x)}h(w,1-a,x)p(w\mid 1-a,x)p(1-a\mid x)p(x).
\end{align*}	
Therefore,
\begin{align*}
&\E[I(A=a)Y]+\sum_{x}\min_w\frac{p(w\mid a,x)}{p(w\mid 1-a,x)}\sum_wh(w,1-a,x)p(w\mid 1-a,x)p(1-a\mid x)p(x)\\
&\le\E[Y^{(A=a)}]\le\\
&\E[I(A=a)Y]+\sum_{x}\max_w\frac{p(w\mid a,x)}{p(w\mid 1-a,x)}\sum_wh(w,1-a,x)p(w\mid 1-a,x)p(1-a\mid x)p(x).
\end{align*}
Note that Assumption \ref{ass:hexistsFD} implies that almost surely
\[
\E[Y\mid A,X]=\E[h(W,A,X)\mid A,X].
\]
Therefore, 
\begin{align*}
&\E[I(A=a)Y]+\sum_{x,y}\min_w\frac{p(w\mid a,x)}{p(w\mid 1-a,x)}yp(y,1-a,x)\\
&\le\E[Y^{(A=a)}]\le\\
&\E[I(A=a)Y]+\sum_{x,y}\max_w\frac{p(w\mid a,x)}{p(w\mid 1-a,x)}yp(y,1-a,x),
\end{align*}	
Combining the above bounds with the trivial bounds based on the support of $Y$ leads to the desired result.

\end{proof}

\begin{proof}[Proof of Corollary \ref{cor:approx}]
For part (a), the LSE inequalities give, almost surely,
\[
r_W(X;-\alpha)\le\min_{w\in\mathcal{W}}v_{w,X},
\qquad
r_W(X;\alpha)\ge\max_{w\in\mathcal{W}}v_{w,X}.
\]
Because $Y\ge0$ almost surely, $\mu_a(X)\ge0$, and hence
\begin{align*}
C(-\alpha)
&\le\E\left[\min_w v_{w,X}\mu_a(X)\mid A=1-a\right],\\
C(\alpha)
&\ge\E\left[\max_w v_{w,X}\mu_a(X)\mid A=1-a\right].
\end{align*}
Theorem \ref{thm:ETTsingleW} implies
\begin{align*}
\max\left\{\inf\mathcal{Y},
\E[\min_wv_{w,X}\mu_a(X)\mid A=1-a]\right\}
&\le\E[Y^{(A=a)}\mid A=1-a],\\
\E[Y^{(A=a)}\mid A=1-a]
&\le
\min\left\{\sup\mathcal{Y},
\E[\max_wv_{w,X}\mu_a(X)\mid A=1-a]\right\}.
\end{align*}
Finally, for any real numbers $x$ and $y$,
\[
LSE(\{x,y\};\alpha)-\frac{\log2}{\alpha}\le\max\{x,y\}
\]
and
\[
LSE(\{x,y\};-\alpha)+\frac{\log2}{\alpha}\ge\min\{x,y\}.
\]
Applying these two inequalities to $(\inf\mathcal{Y},C(-\alpha))$ and
$(\sup\mathcal{Y},C(\alpha))$, respectively, proves part (a).

For part (b),
\[
r_Z(X;-\alpha)\le\min_{z\in\mathcal{Z}}\mu_{a,z}(X),
\qquad
r_Z(X;\alpha)\ge\max_{z\in\mathcal{Z}}\mu_{a,z}(X).
\]
Taking conditional expectations given $A=1-a$ and applying Theorem
\ref{thm:ETTsingleZ} proves the result.
\end{proof}

\begin{proof}[Proof of Corollary \ref{cor:approxalt}]
Since $u_{w,X}=v_{w,X}\mu_a(X)$,
\[
t_W(X;-\alpha)\le\min_{w\in\mathcal{W}}u_{w,X},
\qquad
t_W(X;\alpha)\ge\max_{w\in\mathcal{W}}u_{w,X}.
\]
Therefore,
\begin{align*}
D(-\alpha)
&\le\E[\min_wv_{w,X}\mu_a(X)\mid A=1-a],\\
D(\alpha)
&\ge\E[\max_wv_{w,X}\mu_a(X)\mid A=1-a].
\end{align*}
Combining these inequalities with Theorem \ref{thm:ETTsingleW} and the
two-point LSE inequalities used in the preceding proof yields the claim.
\end{proof}

\begin{proof}[Proof of Corollary \ref{cor:approxmarginal}]
For part (a), the definition of the positive soft extremum and the LSE
inequalities imply
\[
s_W(X;\alpha)\ge\max_{w\in\mathcal{W}}q_w(X)
\]
and
\[
0<s_W(X;-\alpha)\le\min_{w\in\mathcal{W}}q_w(X).
\]
Equivalently, these facts follow directly from
\[
s_W(X;\alpha)
=\left\{\sum_wq_w(X)^\alpha\right\}^{1/\alpha},
\qquad
s_W(X;-\alpha)
=\left\{\sum_wq_w(X)^{-\alpha}\right\}^{-1/\alpha}.
\]
Since $Y\ge0$,
\begin{align*}
C_3(\alpha)
&\le
\E\left[\frac{I(A=a)Y}{\max_wq_w(X)}\right],\\
C_3(-\alpha)
&\ge
\E\left[\frac{I(A=a)Y}{\min_wq_w(X)}\right].
\end{align*}
Corollary \ref{cor:ETTtoATE} now gives
\[
\max\{C_1,C_3(\alpha)\}
\le\E[Y^{(A=a)}]
\le\min\{C_2,C_3(-\alpha)\}.
\]
Applying the two-point LSE inequalities proves part (a). Notice that the
strict positivity of $s_W(X;-\alpha)$ is what justifies taking its reciprocal;
unlike the ordinary additive LSE soft minimum, the positive soft extremum
has this property for every $\alpha>0$.

For part (b), consistency gives
\[
\E[Y^{(A=a)}]
=p(A=1-a)\E[Y^{(A=a)}\mid A=1-a]+\E[I(A=a)Y].
\]
Applying part (b) of Corollary \ref{cor:approx} to the conditional mean in
this display proves the result.
\end{proof}

For the remaining proofs, let $\{P_\varepsilon:
\varepsilon\in(-\delta,\delta)\}$ be a regular parametric submodel through
the true observed-data law $P=P_0$, and let $S(O)$ be its score at
$\varepsilon=0$. A mean-zero square-integrable function $\phi_\Theta(O)$ is
an influence function for $\Theta(P)$ if
\[
\left.\frac{\partial}{\partial\varepsilon}
\Theta(P_\varepsilon)\right|_{\varepsilon=0}
=\E[\phi_\Theta(O)S(O)].
\]
Because the observed-data model is nonparametric, an influence function in
the nonparametric tangent space is the efficient influence function.

\begin{proof}[Proof of Theorem \ref{thm:eifapprox}]
We first consider the $W$-based functionals. For $c\in\{a,1-a\}$, the
pathwise derivatives of the relevant conditional nuisance functions satisfy
\begin{align*}
\dot\mu_a(X)
&=\E\left[
\frac{I(A=a)}{p(a\mid X)}\{Y-\mu_a(X)\}S(O)
\ \middle|\ X\right],\\
\dot p_c(w\mid X)
&=\E\left[
\frac{I(A=c)}{p(c\mid X)}
\{I(W=w)-p_c(w\mid X)\}S(O)
\ \middle|\ X\right].
\end{align*}
Here and below, a dot denotes the derivative along the submodel at zero.
The derivative of the LSE map is
\[
\frac{\partial r_W(X;\gamma)}{\partial v_{w,X}}
=\omega_w^W(X;\gamma),
\]
where
\[
\dot v_{w,X}
=\frac{\dot p_{1-a}(w\mid X)}{p_a(w\mid X)}
-\frac{p_{1-a}(w\mid X)}{p_a(w\mid X)^2}
\dot p_a(w\mid X).
\]
Writing
\[
C(\gamma)
=\E\left[
\frac{p(1-a\mid X)}{p(A=1-a)}
r_W(X;\gamma)\mu_a(X)\right],
\]
the derivative of the outer conditional law contributes
\[
\E\left[
\frac{I(A=1-a)}{p(A=1-a)}
\{r_W(X;\gamma)\mu_a(X)-C(\gamma)\}S(O)\right].
\]
The remaining derivative is
\[
\E\left[
\frac{p(1-a\mid X)}{p(A=1-a)}
\left\{r_W(X;\gamma)\dot\mu_a(X)
+\mu_a(X)\sum_w\omega_w^W(X;\gamma)\dot v_{w,X}\right\}
\right].
\]
Substituting the nuisance derivatives into this display and applying
iterated expectations gives
\[
\dot C(\gamma)=\E[\eta_{C,\gamma}(O)S(O)].
\]
The two outer LSE derivatives are
\[
\frac{\partial\psi_{LW}(\alpha)}{\partial C(-\alpha)}
=L_1(\alpha),
\qquad
\frac{\partial\psi_{UW}(\alpha)}{\partial C(\alpha)}
=L_2(\alpha).
\]
The chain rule therefore gives the two $W$-based EIFs stated in the theorem.

For the $Z$-based functionals,
\[
\dot\mu_{a,z}(X)
=\E\left[
\frac{I(A=a,Z=z)}{p(A=a,Z=z\mid X)}
\{Y-\mu_{a,z}(X)\}S(O)
\ \middle|\ X\right]
\]
and
\[
\frac{\partial r_Z(X;\gamma)}{\partial\mu_{a,z}(X)}
=\omega_z^Z(X;\gamma).
\]
Differentiating
\[
\psi_Z(\gamma)
=\E[r_Z(X;\gamma)\mid A=1-a]
\]
gives one contribution from the outer conditional distribution and one from
the nuisance regressions:
\begin{align*}
\dot\psi_Z(\gamma)
={}&\E\left[
\frac{I(A=1-a)}{p(A=1-a)}
\{r_Z(X;\gamma)-\psi_Z(\gamma)\}S(O)\right]\\
&+\E\left[
\frac{p(1-a\mid X)}{p(A=1-a)}
\sum_z\omega_z^Z(X;\gamma)\dot\mu_{a,z}(X)\right].
\end{align*}
Substitution of the displayed derivative of $\mu_{a,z}$ yields
\[
\dot\psi_Z(\gamma)
=\E[\eta_{\psi_Z,\gamma}(O)S(O)],
\]
which proves the result after setting $\gamma=-\alpha$ and $\gamma=\alpha$.
\end{proof}

\begin{proof}[Proof of Theorem \ref{thm:eifapproxalt}]
The derivative of
\[
D(\gamma)=\E[t_W(X;\gamma)\mid A=1-a]
\]
with respect to its outer conditional law contributes
\[
\E\left[
\frac{I(A=1-a)}{p(A=1-a)}
\{t_W(X;\gamma)-D(\gamma)\}S(O)\right].
\]
Moreover,
\[
\frac{\partial t_W(X;\gamma)}{\partial u_{w,X}}
=\widetilde\omega_w^W(X;\gamma)
\]
and
\[
\dot u_{w,X}
=v_{w,X}\dot\mu_a(X)+\mu_a(X)\dot v_{w,X}.
\]
Consequently, the derivative through the nuisance functions equals
\[
\E\left[
\frac{p(1-a\mid X)}{p(A=1-a)}
\sum_w\widetilde\omega_w^W(X;\gamma)
\{v_{w,X}\dot\mu_a(X)+\mu_a(X)\dot v_{w,X}\}
\right].
\]
Using the nuisance derivatives displayed in the proof of Theorem
\ref{thm:eifapprox} and applying iterated expectations gives
\[
\dot D(\gamma)=\E[\eta_{D,\gamma}(O)S(O)].
\]
Finally,
\[
\frac{\partial\widetilde\psi_{LW}(\alpha)}{\partial D(-\alpha)}
=\widetilde L_1(\alpha),
\qquad
\frac{\partial\widetilde\psi_{UW}(\alpha)}{\partial D(\alpha)}
=\widetilde L_2(\alpha).
\]
The chain rule proves the theorem.
\end{proof}

\begin{proof}[Proof of Theorem \ref{thm:eifmarginal}]
For the $W$-based bounds, the EIFs of $C_1$ and $C_2$ are
\[
\eta_1(O)=\inf\mathcal{Y}\,I(A=1-a)+I(A=a)Y-C_1
\]
and
\[
\eta_2(O)=\sup\mathcal{Y}\,I(A=1-a)+I(A=a)Y-C_2,
\]
respectively. It remains to differentiate $C_3(\gamma)$. First, the positive
soft extremum satisfies
\[
\frac{\partial s_W(X;\gamma)}{\partial q_w(X)}
=s_W(X;\gamma)\frac{\lambda_w(X;\gamma)}{q_w(X)}.
\]
The pathwise derivative of $q_w(X)=p(A=a\mid W=w,X)$ is
\[
\dot q_w(X)
=\E\left[
\frac{I(W=w)}{p(w\mid X)}
\{I(A=a)-q_w(X)\}S(O)
\ \middle|\ X\right].
\]
The direct variation of the outer law in
\[
C_3(\gamma)
=\E\left[\frac{I(A=a)Y}{s_W(X;\gamma)}\right]
\]
contributes
\[
\E\left[
\left\{\frac{I(A=a)Y}{s_W(X;\gamma)}-C_3(\gamma)\right\}S(O)
\right].
\]
The chain rule for the reciprocal of $s_W$ contributes
\[
-\E\left[
\frac{\E[I(A=a)Y\mid X]}{s_W(X;\gamma)^2}
\sum_w\frac{\partial s_W(X;\gamma)}{\partial q_w(X)}
\dot q_w(X)\right].
\]
Substituting the derivatives of $s_W$ and $q_w$ shows that
\[
\dot C_3(\gamma)=\E[\eta_{3,\gamma}(O)S(O)].
\]
In particular, the propensity-score correction has a minus sign because
the functional contains $1/s_W$.

The derivatives of the two outer LSE maps with respect to their first
arguments are $K_1(\alpha)$ and $K_2(\alpha)$; their derivatives with respect
to their second arguments are $1-K_1(\alpha)$ and $1-K_2(\alpha)$.
Combining the EIFs of $C_1$, $C_2$, and $C_3$ proves the $W$-based claims.

For the $Z$-based bounds, write
\[
\varphi_Z(\gamma)
=M_Z(\gamma)+\E[I(A=a)Y]
=\E[I(A=1-a)r_Z(X;\gamma)]+\E[I(A=a)Y].
\]
The direct variation of the two outer expectations contributes
\[
I(A=1-a)r_Z(X;\gamma)-M_Z(\gamma)
+I(A=a)Y-\E[I(A=a)Y].
\]
The only additional variation is through $\mu_{a,z}(X)$. Using
\[
\frac{\partial r_Z(X;\gamma)}{\partial\mu_{a,z}(X)}
=\omega_z^Z(X;\gamma)
\]
and the nuisance derivative from the proof of Theorem
\ref{thm:eifapprox}, this contribution is represented by
\[
\sum_{z\in\mathcal{Z}}
\frac{I(A=a,Z=z)}{p(A=a,Z=z\mid X)}
p(1-a\mid X)\omega_z^Z(X;\gamma)
\{Y-\mu_{a,z}(X)\}.
\]
Thus
\[
\dot\varphi_Z(\gamma)
=\E[\eta_{\varphi_Z,\gamma}(O)S(O)],
\]
and setting $\gamma=-\alpha$ and $\gamma=\alpha$ gives the lower and upper
EIFs, respectively.
\end{proof}

\end{document}